\def\final{0}
\def\notes{0}
\def\informal{1}
\def\lip{1}
\def\appen{1}

\ifnum\final=0
\documentclass[11pt]{article}
\else
\documentclass[prodmode,license]{acmsmall-ec15}
\usepackage{url}
\fi

\ifnum\final=0
\usepackage{amsmath, amssymb, amsthm}
\usepackage{natbib}
\else
\usepackage{amsmath, amssymb}
\fi

\usepackage{mathtools}
\usepackage{framed}
\ifnum\final=0
\usepackage{fullpage}
\usepackage[margin=1.0in]{geometry}
\fi

\usepackage{verbatim}
\usepackage{color}
\usepackage{caption}
\definecolor{DarkGreen}{rgb}{0.1,0.5,0.1}
\definecolor{DarkRed}{rgb}{0.5,0.1,0.1}
\definecolor{DarkBlue}{rgb}{0.1,0.1,0.5}

\ifnum\notes=1
\newcommand{\mynote}[1]{\marginpar{\tiny\sf #1}}
\else
\newcommand{\mynote}[1]{}
\fi
\newcommand{\jnote}[1]{\mynote{\color{DarkGreen}Jon: {#1}}}
\newcommand{\snote}[1]{\mynote{\color{blue}Steven: {#1}}}
\newcommand{\rynote}[1]{\mynote{\color{red}Ryan: {#1}}}
\newcommand{\anote}[1]{\mynote{\color{DarkRed}Aaron: {#1}}}

\newcommand{\Prob}[2]{\underset{#1}{\mathbb{P}}\left[ #2 \right]}
\newcommand{\Exp}[2]{\underset{#1}{\mathbb{E}}\left[ #2 \right]}

\newcommand{\bs}{\mathbf{s}}
\newcommand{\bby}{\mathbf{y}}
\newcommand{\bbx}{\mathbf{x}}
\newcommand{\bbz}{\mathbf{z}}
\newcommand{\utt}{^{(t)}}
\newcommand{\bl}{^\bullet}

\newcommand{\eps}{\varepsilon}
\newcommand{\Lap}{\text{Lap}}
\newcommand{\pigd}{\texttt{P-GD}}
\newcommand{\psrr}{\texttt{PSRR}}

\newcommand{\sbr}{\texttt{P-BR}}
\newcommand{\med}{\texttt{FlowToll}}
\newcommand{\gd}{\texttt{GD}}
\newcommand{\bt}{\mathcal{M}_{\text{BT}}}
\newcommand{\ML}{\mathcal{M}_L}
\newcommand{\pcon}{\texttt{P-CON}}

\newcommand\N{\mathbb{N}}
\newcommand\R{\mathbb{R}}

\newcommand{\cA}{\mathcal{A}}
\newcommand{\cB}{\mathcal{B}}

\newcommand{\cD}{\mathcal{D}}

\newcommand{\cF}{\mathcal{F}}
\newcommand{\cG}{\mathcal{G}}

\newcommand{\cL}{\mathcal{L}}
\newcommand{\cM}{\mathcal{M}}

\newcommand{\cR}{\mathcal{R}}
\newcommand{\cS}{\mathcal{S}}

\newcommand{\cX}{\mathcal{X}}

\newcommand{\polylog}{\mathrm{polylog}}
\newcommand{\bits}{\{0,1\}}

\newcommand{\from}{:}

\newcommand{\argmin}{\arg\!\min}
\newcommand{\argmax}{\arg\!\max}
\newcommand{\OPT}{\text{OPT}}

\ifnum\final=0

\usepackage{algorithm}
\usepackage{algpseudocode}
\algtext*{EndWhile}
\algtext*{EndFor}
\algtext*{EndIf}
\algnewcommand\algorithmicinput{\textbf{Input:}}
 \algnewcommand\INPUT{\item[\algorithmicinput]}
 \algnewcommand\algorithmicoutput{\textbf{Output:}}
 \algnewcommand\OUTPUT{\item[\algorithmicoutput]}

\else
\usepackage[ruled, noend]{algorithm2e}

\SetAlFnt{\small}
\SetAlCapFnt{\small}
\SetAlCapNameFnt{\small}
\SetAlCapHSkip{0pt}
\IncMargin{-\parindent}

\doi{2764468.2764509}

\acmVolume{X}
\acmNumber{X}
\acmArticle{X}
\acmYear{2015}
\acmMonth{2}
\usepackage[numbers, sort&compress]{natbib}

\fi

\DeclareMathOperator*{\myargmin}{\mathrm{argmin}}
\ifnum\final=0
\newtheorem{theorem}{Theorem}[section]
\newtheorem{lemma}[theorem]{Lemma}

\newtheorem{remark}[theorem]{Remark}
\newtheorem{corollary}[theorem]{Corollary}

\theoremstyle{definition}
\newtheorem{definition}[theorem]{Definition}
\fi

\newtheorem{assumpt}[theorem]{Assumption}

\usepackage{hyperref}

\ifnum\final=0
\usepackage{graphicx}
\fi

\hypersetup{
    unicode=false,				
    pdftoolbar=true,				
    pdfmenubar=true,			
    pdffitwindow=false, 		
    pdftitle={},    					
    pdfauthor={}
    pdfsubject={},				
    pdfnewwindow=true,		
    pdfkeywords={},			
    colorlinks=true,				
    linkcolor=DarkGreen,		
    citecolor=DarkGreen,	
    filecolor=DarkRed,		
    urlcolor=DarkBlue,		
}

\ifnum\final=1
\begin{document}
\fi

\title{Inducing Approximately Optimal Flow Using Truthful Mediators}

\ifnum\final=0
\author{Ryan Rogers \and Aaron Roth\thanks{University of Pennsylvania Department of Computer and Information Sciences. Partially supported by an NSF CAREER award, NSF Grants CCF-1101389 and CNS-1065060, and a Google Focused Research Award. \href{mailto:aaroth@cis.upenn.edu}{aaroth@cis.upenn.edu}} \and Jonathan Ullman\thanks{Columbia University Department of Computer Science.  Supported by a Junior Fellowship from the Simons Society of Fellows.  Email: \href{mailto:jullman@cs.columbia.edu}{jullman@cs.columbia.edu}.} \and Zhiwei Steven Wu}
\else
\markboth{Rogers et al.}{Inducing Approximately Optimal Flow Using Truthful Mediators}
\author{RYAN ROGERS
\affil{University of Pennsylvania}
AARON ROTH
\affil{University of Pennsylvania}
JONATHAN ULLMAN
\affil{Columbia University}
ZHIWEI STEVEN WU
\affil{University of Pennsylvania}
}
\terms{Mechanism Design, Differential Privacy, Algorithms}


\acmformat{Ryan Rogers, Aaron Roth, Jonathan Ullman and Zhiwei Steven Wu, 2015. Inducing Approximately Optimal Flow Using Truthful Mediators.}

\begin{bottomstuff}
This work is supported in part by NSF grants CCF-1101389 and CNS-1253345.

Author's addresses: Ryan Rogers, Applied Math and Computational Science Department,
University of Pennsylvania; Aaron Roth,
Computer and Information Science Department, University of Pennsylvania; Jonathan Ullman, Department of Computer Science, Columbia University; Zhiwei Steven Wu, Computer and Information Science Department, University of Pennsylvania.

\end{bottomstuff}
\fi

\ifnum\final=0
\begin{document}
\maketitle
\fi

\begin{abstract}
We revisit a classic coordination problem from the perspective of mechanism
design: how can we coordinate a social welfare maximizing flow in a network
congestion game with selfish players? The classical approach, which
computes tolls as a function of known demands, fails when the demands are
unknown to the mechanism designer, and naively eliciting them does not
necessarily yield a truthful mechanism. Instead, we introduce a \emph{weak
mediator} that can provide suggested routes to players and set tolls as a
function of reported demands.  However, players can choose to ignore or
misreport their type to this mediator.  Using techniques from differential
privacy, we show how to design a weak mediator such that it is an
asymptotic ex-post Nash equilibrium for all players to truthfully report
their types to the mediator and faithfully follow its suggestion, and that
when they do, they end up playing a nearly optimal flow. Notably, our solution
works in settings of incomplete information even in the absence of a prior
distribution on player types. Along the way, we develop new techniques for
privately solving convex programs which may be of independent interest.
\end{abstract}

\maketitle


\ifnum\final=0
\vfill
\newpage
\tableofcontents

\vfill
\newpage

\fi

\section{Introduction}
Large, atomic traffic routing games model the common scenario in which $n$ agents
(say, residents of a city) must choose paths in some graph (the road network)
to route a unit of flow (drive to work) between their target source/sink
pairs. In aggregate, the decisions of each of these agents cause congestion
on the edges (traffic), and each agent experiences a cost equal to the sum of
the latencies of the edges she traverses, given the decisions of everyone
else. The latencies on each edge are a function of the congestion on that
edge.

This widely studied class of games presents several well known challenges:
\begin{enumerate}
\item First, for the social welfare objective, the price of anarchy is
    unboundedly large when the latencies can be arbitrary convex functions.

\item Second, in atomic routing games, equilibria are not unique, and hence
    equilibrium selection is an important problem.

\item Finally, as in most large games, players will be generally unaware of
    the types of their opponents, and so it is important to understand these games
    in settings of incomplete information.
\end{enumerate}

One way to address the first challenge is to introduce carefully selected \emph{tolls} on the
edges, which modifies the game and decreases the price of
anarchy. Indeed, so called \emph{marginal cost tolls} make the socially
optimal routing a Nash equilibrium.  The marginal cost toll on each edge
charges each agent the cost that she imposes on all other agents.
However, in atomic congestion games with marginal cost tolls, the socially
optimal routing is not necessarily the only Nash equilibrium routing, and so
the price of anarchy can be larger than 1, and the coordination problem is
still not solved.  Moreover, because it is difficult to charge agents tolls
\emph{as a function of what others are doing} (as the marginal cost tolls
do), there is a large literature that considers the problem of finding
\emph{fixed} tolls that induce the optimal routing, under various conditions
\cite{CDR03,FJM04,KK04,Fle05,Swa07,FKK10}

This literature, however, assumes the agents' source/sink pairs are known,
and computes the tolls as a function of this information. In this paper we
instead take a mechanism design approach---the demands of the agents must be
elicited, and agents may misrepresent their demands if it is advantageous to
do so. Compared to standard mechanism design settings, our mechanism is
somewhat restricted: it can only set anonymous tolls, and cannot require
direct payments from the agents, and it also cannot force the agents to take
any particular route. Because of these limitations, standard tools like the
VCG mechanism do not apply. Instead, we approach the problem by introducing a
\emph{weak mediator} which also solves the 2nd and 3rd problems identified
above---it solves the equilibrium selection problem, even in settings of
incomplete information. The solutions we give are all \emph{approximate}
(both in terms of the incentives we guarantee, and our approximation to the
optimal social welfare), but the solution approaches perfect as the game
grows large.

Informally, a weak mediator is an intermediary with whom agents can choose to
interact with.  This leads to a new \emph{mediated game}, related to the original routing game.
In our setting, the weak mediator elicits the
types of each agent. Based on the agents reports, it fixes constant tolls to
charge on each edge, and then suggests a route for each agent to play. However, agents
are free to act independently of the mediator.  They need not report
their type to it honestly, or even report a type at all. They are also not obligated
to follow the route suggested by the mediator, and can deviate from it in
arbitrary ways. Our goal is to design a mediator that incentivizes ``good
behavior'' in the mediated game---that agents should truthfully report their type to the mediator,
and then faithfully follow its suggestion. Moreover, we want that when agents
do this, the resulting routing will be socially optimal.

\jnote{Would it make sense to have a more explicit (but still informal) statement of our main theorem in the intro?  As it is we don't get to the main theorem until page 10.}
Our main result is that this is possible in large routing games with convex
loss functions. By large, we mean both that the number of players $n$ is
large, and that the latency functions are Lipschitz continuous---i.e.~that no
single agent can substantially affect the latency of any edge via a
unilateral deviation. We give a weak mediator that makes ``good behavior'' an
approximate \emph{ex-post} Nash equilibrium---i.e.~a Nash equilibrium in
\emph{every} game that might be induced by realizations of the agents types.
This is an extremely robust solution concept that applies even when agents
have no distributional knowledge of each other's types.  In the limit as $n$
goes to infinity, the approximate equilibrium becomes exact. The mediator
also implements an approximately optimal routing, in that the welfare of the
suggested routing is suboptimal by an additive term that is sublinear in $n$.
Hence, if the cost of the optimal routing grows linearly, or nearly linearly
in $n$, then the approximately optimal flow achieves a fraction of the
optimal social welfare that is arbitrarily close to $1$.

\subsection{Our Techniques and Main Results}
At a high level, the approach we take is to design a mediator which takes as
input the reported source/destination pairs of each agent, and as a function
of those reports:
\begin{enumerate}
\item Computes the optimal routing given the reported demands, and
\item Computes fixed tolls that make this routing a Nash
    equilibrium, and finally
\item Suggests to each player that they play their part of this optimal
    routing.
\end{enumerate}
However, implementing each of these steps straightforwardly does not make good behavior an equilibrium in general. Agents may hope to gain in two
ways by misreporting their type: they may hope to change the tolls charged on
the path that they eventually take, and they may hope to change the
algorithm's suggestions to other players, to change the edge congestions.
Simply because the game is large, and hence each player has little direct
effect on the costs of other players does not necessarily mean that no
player's report can have large effect on an \emph{algorithm} which is
computing an equilibrium (see e.g. \cite{KPRU14} for an example).

To address this problem, we follow the approach taken in \cite{KPRU14,
RR14} and compute the optimal routing and tolls using \emph{joint
differential privacy}.  Informally, joint differential privacy guarantees
that if any agent unilaterally misreports her demand, then it has only a
small effect on the routes taken by \emph{every other agent}, as well as
on the tolls. (It of course has a very large effect on the route suggested to
that agent herself, since she is always given a route between her reported
source/sink pairs!) As we show, this is sufficient to guarantee that an agent
cannot benefit substantially by misreporting her demand. Assuming the other
agents behave honestly---meaning they report their true demand and follow
their suggested route---then the fact that the algorithm \emph{also} is
guaranteed to compute a routing which forms an approximate equilibrium of the
game, given the tolls, guarantees that agents cannot do substantially better
than also playing honestly, and playing their part of the computed
equilibrium.

In order to do this, we need to develop new techniques for convex
optimization under joint differential privacy. In particular, in order to
find the socially optimal flow privately, we need the ability to privately
solve a convex program with an objective that is not linearly separable
among players, and hence one for which existing techniques \cite{HHRW14} do
not apply.

We now informally state the main theorem of this paper.  It asserts that
there is a mediator 
that incentivizes good behavior as an ex-post Nash equilibrium, while
implementing the optimal flow. Here we assume that the latency functions on the edges are bounded by the number of players $n$ and are \ifnum\lip=0$\Theta(1)$-\fi Lipschitz continuous\ifnum\lip=0,\footnote{We believe this is the most natural parameter regime for a large game.  If the latency incurred when all $n$ players use an edge is at most $n$, then the average contribution of each player to the congestion is at most $1$.  Assuming that the latency is $\Theta(1)$-Lipschitz implies that no player's contribution is more than a constant fraction larger than the average.}\fi  although our formal theorem statement gives more general parameter tradeoffs.
\begin{theorem}[Informal] \label{thm:maininformal}
For large\footnote{The formal notion of largeness we require is detailed in Assumption~\ref{assumption}.} routing games with $n$ players and $m$ edges, there exists a mediator $M$ such that good behavior is an $\eta_{eq}$-approximate Nash equilibrium in the mediated game where
\ifnum\final=1
$
\else
$$
\fi
\eta_{eq} = \tilde O\left( m^{3/2} n^{4/5} \right)
\ifnum\final=1
$
\else
$$
\fi
and when players follow good behavior, the resulting flow is an $\eta_{opt}$-approximately optimal average flow for the original routing game where
\ifnum\final=1
$
\else
$$
\fi
\eta_{opt} = \tilde O\left( m n^{4/5} \right).
\ifnum\final=1
$
\else
$$
\fi
 \label{thm:main-veryinformal}
\end{theorem}

\ifnum\informal=1
To interpret this theorem, let us write $\mathrm{OPT}$ to denote the \emph{average player latency} in the socially optimal flow. Note that in this parameter regime (latency functions which are bounded by $n$ and \ifnum\lip=0 $\Theta(1)$-\fi Lipschitz), if the value OPT increases at a rate faster than $n^{4/5}$ as the population $n$ grows, then our mediator yields a flow that obtains average latency $(1+o_n(1))\cdot \mathrm{OPT}$.\footnote{Here, $o_n(1)$ denotes a function of $n$ that approaches $0$ as $n \rightarrow \infty$.}  We view this condition on OPT as very mild.  For example, if the network is fixed and all of the latency functions have derivatives bounded strictly away from zero, then the optimal average latency will grow at a rate of $\Omega(n)$. Our results hold \emph{even} when the optimal average latency grows sublinearly. Similarly, in this setting, for a $1 - o_n(1)$ fraction of individuals the latency of their best response route also grows at a rate of $\Omega(n)$, and hence our mediator guarantees that for a $(1-o_n(1))$-fraction of individuals, they are playing an $(1-o_n(1))$-approximate best-response (i.e. they cannot decrease their latency by more than a $1-o_n(1)$ multiplicative factor by deviating from the mediator's suggestion).

\else
We note that with the mediator $M$ from the above theorem, as $n \to \infty$, we get an asymptotically exact Nash equilibrium in $\Gamma_M$.  Further, suppose that when we fix $m$, allow the number of players $n$ to grow large and the value OPT of the socially optimal flow grows faster than $n^{4/5}$, then the flow our algorithm outputs has cost $(1+o_n(1))\cdot \mathrm{OPT}$.\footnote{Here, $o_n(1)$ denotes a function of $n$ that approaches $0$ as $n \rightarrow \infty$.}
\fi

\subsection{Related Work}
There is a long history of using tolls to modify the equilibria in congestion
games (see e.g.~\cite{BMW56} for a classical treatment). More recently, there
has been interest in the problem of computing fixed tolls to induce optimal
flows at equilibrium in various settings, usually in \emph{non-atomic}
congestion games (see e.g. \cite{CDR03,FJM04,KK04,Fle05,Swa07,FKK10} for a
representative but not exhaustive sample). These papers study variations on
the problem in which e.g. tolls represent lost welfare \cite{CDR03}, or in
which agents have heterogenous values for money \cite{FJM04}, or when agents
are atomic but flow is splittable \cite{Swa07}, among others. Tolls in atomic
congestion games have received some attention as well (e.g.~\cite{CKK06}),
though to a lesser degree, since in general atomic congestion games, tolls do
not suffice to implement the optimal flow as the unique equilibrium). These
works all assume that agent demands are known, and do not have to be elicited
from strategic agents, which is where the present paper departs from this
literature. Recently, Bhaskar et al. \cite{BLSS14} consider the problem of
computing tolls in a query model in which the latency functions are unknown
(demands are known), but not in a setting in which agents are assumed to be
behaving strategically to manipulate the tolls.

Modifying games by adding ``mediators'' is also well studied,
although what exactly is meant by a mediator differs from paper to paper (see
e.g. \cite{mediators1,mediators2,RT07,AMT09,PP10} for a representative but
not exhaustive sample). The ``weak mediators'' we study in this paper were
introduced in \cite{KPRU14,RR14}, who also use differentially private equilibrium computation to achieve incentive properties. Our work differs
from this prior work in that \cite{KPRU14,RR14} both seek to implement an
equilibrium of the given game, and hence do not achieve welfare guarantees
beyond the price of anarchy of the game. In contrast, we use tolls to modify
the original game, and hence implement the socially optimal routing as an
equilibrium.

The connection between differential privacy, defined by \cite{DMNS06},
and mechanism design was first made by \cite{MT07}, who used it to give improved welfare guarantees for
digital goods auctions. It has since been used in various contexts,
including to design mechanisms for facility location games and general
mechanism design problems without money \cite{NST12}. The connection
between \emph{joint} differential privacy and mechanism design (which
is more subtle, and requires that the private algorithm also compute
an equilibrium of some sort) was made by \cite{KPRU14} in the context
of mediators, and has since been used in other settings including
computing stable matchings \cite{KMRW15}, aggregative
games~\cite{CKRW14}, and combinatorial auctions \cite{HHRW14}.

\section{Model}
\subsection{The Routing Game Problem}
In this section we introduce the atomic unsplittable routing game
problem that we study.  An instance of a routing game $\Gamma =
(G,\ell,\bs)$ is defined by
\begin{itemize}
\item[$\bullet$] A graph $G = (V,E)$.  We use $m = |E|$ to denote the number of edges.
\item[$\bullet$] A \emph{latency function} $\ell_e \from \R_{\geq 0} \to \R_{\geq 0}$ for each edge $e \in E$.  Each latency function maps the number of players who send flow along that edge to a non-negative loss.
\item[$\bullet$] A set of $n$ source-destination pairs $\bs = (s_1,\dots,s_n)$.  Each pair $s_i = (s_i^1, s_i^2) \in \cS\equiv V \times V$ represents the \emph{demand of player $i$}.  We use $n$ to denote the number of players.
\jnote{We should go through and try to be consistent about player/agent, unless using the two interchangeably is standard in AGT.}
\rynote{Will try to stick with players}
\end{itemize}

The objective is to (approximately) minimize the total latency
experienced by all the players in the network.  Let
$\cF(\bs)=(\cF(s_1),\cdots,\cF(s_n))$ be the set of \emph{feasible
  individual flows for demand $\bs$} and $\cF = \{\cF(\bs): \bs \in
\cS\}$ be the set of all feasible individual flows.\jnote{Try to be
  consistent about how the prose relates to the mathematical objects.
  Previously $\cF(s)$ was the set of ``feasible individual flows''
  (making no reference to $\bs$), and $\cF$ had no name}
  \rynote{Better now?}Notice that
an element of $\cF(\bs)$ is a vector of $n$ separate flows, one for
each player.  That is, an individual flow is specified by $n\times m$
variables representing the amount of flow by each player routed on
each edge.  Specifically, given a graph $G$, $\cF(\bs)$ is the set of unsplittable
flows $\mathbf{x} = (x_{i,e})_{i \in [n], e \in E} \in \bits^{n \times
  m}$ such that
\begin{align}
	& b_{i,u} =\left\{\begin{array}{lr}
      1 & u = s^1_i \\
      -1 & u = s^2_i \\
      0 & \text{ else }
\end{array}\right.  \label{eq:demand}\\
& b_{i,u} + \sum_{v: (u,v) \in E} x_{i,(u,v)} = \sum_{v: (v,u) \in E}
x_{i,(v,u)}, \; \quad \forall u \in V \quad \forall i \in
[n] \label{eq:conservation}
\end{align}
For a given routing game instance $\Gamma = (G,\ell,\bs)$, we seek a flow $\bbx\in \cF(\bs)$ that minimizes the average latency $\phi(\mathbf{x})$

\begin{align}
 \phi(\mathbf{x}):=\frac{1}{n}\sum_{i=1}^n \sum_{e \in E} x_{i,e} \cdot \ell_e\left( \sum_{i = 1}^n x_{i,e} \right) \label{eq:min_cost_flow}
\end{align}


We will sometimes write $\OPT(\bs) = \phi(\bbx^*)$, were $\bbx^*$ is the minimum average cost flow for the routing game $\Gamma = (G,\ell, \bs)$ when the graph $G$ and latencies $\ell$ are known from context.  In this work we  settle for an approximately minimum average cost flow, which we define below.

\begin{definition}[Approximately Optimal Flow] For a routing game $\Gamma$, and parameter $\eta_{opt} > 0$, a flow $\bbx$ is $\eta_{opt}$-\emph{approximately optimal} if $\bbx \in \cF(\bs)$ and
\ifnum\final=1
$
\else
$$
\fi
\phi( \bbx) \leq \OPT(\bs) + \eta_{opt}.
\ifnum\final=1
$
\else
$$
\fi
  \end{definition}

We are interested in strategic players that want to minimize their individual cost
\begin{equation}
 \phi_i(\bbx)=\sum_{e \in E}  x_{i,e} \cdot \ell_e\left(\sum_{j=1}^n  x_{j,e}\right).
 \label{eq:ind_cost}
\end{equation}
We thus define an approximate Nash flow.
\begin{definition}[Approximate Nash Flow]
 For a routing game $\Gamma$ and parameter $\eta_{eq}>0$, a flow $\hat \bbx$ is an \emph{$\eta_{eq}$-approximate Nash flow} if $\hat\bbx \in \cF(\bs)$ and for every $x_i \in \cF(s_i)$
 \ifnum\final=1
 $
 \else
 $$
 \fi
 \phi_i(\hat\bbx) \leq \phi_i(x_i,\hat\bbx_{-i})+ \eta_{eq} \qquad \forall x_i \in \cF(s_i).
 \ifnum\final=1
 $
 \else
 $$
 \fi
 When $\hat \bbx$ is a $0$-approximate Nash flow, we simply say that it is a \emph{Nash flow}.
\end{definition}

Throughout, we will make the following assumptions about the latency functions.

\begin{assumpt}
  For every edge $e\in E$, the latency function $\ell_e$ is (1) non-decreasing, (2) convex, (3) twice differentiable, (4) bounded by $n$ (i.e.~$\ell_e(n) \leq n$), and (5) $\gamma$-lipschitz (i.e.~$|\ell_e(y) - \ell_e(y')| \leq \gamma |y
    - y'|$ for all $e \in E$) for some \ifnum\lip=0 $\gamma = o(n)$ \else constant $\gamma>0$\fi. 
    \label{assumption}
\end{assumpt}

Item 1 and 2 are natural and extremely common in the routing games
literature.  Item 3 is a technical condition used in our proofs that
can likely be removed.  Item 4 and 5 are the ``largeness conditions''
that ensure no player has large influence on any other's payoff.  \ifnum\lip=0 Here $\gamma$ is an upper bound on the Lipschitz constant for $\ell_e$ which should not grow with $n$.  Hence, we can take $\gamma = \Theta(1)$.  \else If the Lipschitz constant is zero, then we can choose an upper bound parameter $\gamma >0$ in our analysis.\fi   
\jnote{Added a quick summary of how these assumptions are used.}

\subsection{Mediators}
Given an instance $\Gamma = (G, \ell, \bs)$, we would like the players to
coordinate on the social-welfare maximizing flow $\bbx^*$ where $\OPT(\bs) =
\phi(\bbx^*)$. There are two problems: the first is that the optimal flow is
generally not a Nash equilibrium, and the second is that even with knowledge
of everyone's demands, Nash equilibria are not unique and coordination is
a problem. The classical solution to the first problem is to have an overseer
impose edge tolls $\tau$, which are a function of the demands $\bs$ of each player.  This makes $\bbx^*$ a Nash flow for the routing
game instance $\Gamma^{\tau}= (G,\ell^\tau,\bs)$ where
\ifnum\final=1 $ \else $$ \fi
\ell_e^\tau(y) = \ell_e(y) + \tau_e.
\ifnum\final=1 $ \else $$ \fi

However the tolls that cause the optimal flow to be an equilibrium depend on the demands, and so this approach
fails if the overseer does not know $\bs$.  \jnote{We have already
defined ``demands'' in the technical sections and used it several times.  We
use ``types'' in the introduction but not in the technical section.  Can we
just pick one?  I prefer demand since it's more descriptive.}\snote{I vote for demands as well}\rynote{Will verify demands are consistent} A simple
solution would be to elicit the demands from the players, but since the
correct tolls depend on the demands, naively eliciting them may not
lead to a truthful mechanism.

We solve this problem, as well as the equilibrium selection problem mentioned
above, by introducing a \emph{mediator} that takes as input the demand of each
player and outputs a set of tolls for each edge, together with a suggested
route for each player to use. Ideally, the players will report their
demands truthfully, the aggregate of the routes suggested by the mediator
will be a social-welfare maximizing flow $\bbx^*$, agents will faithfully
follow their suggestion, and the tolls will be chosen to make $\bbx^*$ an
(approximate) Nash flow. However, players have the option to deviate from
this desired behavior in several ways: they may not report their demand to the
mediator at all, might report a false demand, or might not follow the
mediator's suggestion once it is given. Our goal in designing the mediator is
to guarantee that players never have significant incentive to
deviate from the desired behavior described above. \jnote{I modified this section a bit.
Someone should check it.}\rynote{Took out mention of ex-post Nash here because we do not define until several paragraphs later}

Formally, introducing the mediator gives rise to a modified game $\Gamma_M =
(G,\ell,\bs,M)$.   The mediator is an algorithm $M:\{\bot\cup\cS  \}^n \to
\cF^n \times \R^m $.  The input from each player is either a demand or a $\bot$
symbol indicating that the player opts out.  The output is a set of routes,
one suggested to each player, together with a collection of tolls, one for
each edge. We write the output as
\ifnum\final=1 $ \else $$ \fi
M(\bs) = \left(\left(M^\cF_i(\bs)\right)_{i\in[n]},M^\tau(\bs)\right).
\ifnum\final=1 $ \else $$ \fi
The edge tolls $M^\tau(\bs) = (M_e^\tau(\bs))_{e \in E}$ that $M$ outputs will
enforce the optimal flow induced by the reported demands.  Note that the
tolls that $M$ outputs to each player are the same (i.e. the players are not
charged personalized tolls; rather there is a single toll on each edge that
must be paid by any player using that edge).

In $\Gamma_M$ each player can \emph{opt-out} of using the mediator, denoted
by the report $\bot$, and then select some way to route from
his source to his destination, or a player can \emph{opt-in} to using the
mediator, but not necessarily reveal her true demand, and then the mediator
will suggest a path $\bbx_i$ to route her unit flow from the reported source
to the destination. Players are free to follow the suggested action, but they
can also use the suggestion as part of an arbitrary deviation, i.e. they
can play any action $f(\bbx_i)$ for any $f: \cF \to \cF$. Thus, the action
set $\cA$ for any player for the game instance $\Gamma_M$ is $\cA = A_1 \cup A_2$ where
$
A_1 = \{ (s',f): s' \in \cS,f:\cF \to \cF\} \quad \text{and} \quad A_2 = \{ (\bot,f) : f \text{ constant }\}.
$

We next define the cost function for each player in $\Gamma_M$, but first we must present some notation.  Let
$\mathbf{F}$ be the set of possible functions $f_i:\cF \to \cF$, where
$f_i(\bbx_i) = (f_{i,e}(x_{i,e}))_{e\in E}$.  We further write
$f_e(\bbx) = \sum_{i=1}^n f_{i,e}(x_{i,e})$ as the new congestion on edge $e$ when players have deviated from $\bbx$ according to functions $f_i$ for $i \in [n]$. We will consider only
randomized algorithms, so our cost is an expectation over outcomes of
$M$.  More formally, the cost $\phi^M$ that each player experiences in
$\Gamma_M$ is defined as
\begin{align*}
&\phi^M: \cS \times \left[(\bot \cup \cS) \times \mathbf{F} \right]^n
\to \R \\ 
&\phi^M(s_i,(\bs',\mathbf{f}) ) := \Exp{(\bbx,\tau) \sim
  M(\bs')}{\sum_{e \in E}
  f_{i,e}(x_{i,e})\left(\underbrace{\ell_e\left(f_e(\bbx)\right)+ \tau_e}_{\ell_e^\tau(f_e(\bbx))}\right)}
\end{align*}
where $s_i$ is player $i$'s true source-destination pair.

We are interested in designing mediators such that \emph{good
behavior} in the mediated game is an ex-post Nash equilibrium, which we define below.

\begin{definition}[Ex-Post Nash Equilibrium]\label{def:expost}
A set of strategies $\{\sigma_i\colon \cS \to \cA\}_{i=1}^n$ forms an
$\eta$-approximate ex-post Nash equilibrium if for every profile of demands
$\bs\in \cS^n$, and for every player $i$ and action $a_i\in \cA$:
\[
\phi^M\left(s_i, (\sigma_i(s_i), \sigma_{-i}(s_{-i})) \right) \leq
\phi^M\left(s_i, (a_i, \sigma_{-i}(s_{-i})) \right) + \eta.
\]
That is, it forms an $\eta$-approximate Nash equilibrium for every
realization of demands.
\end{definition}
\snote{define ex-post (maybe elaborate more) and good behavior here}
\rynote{I think it's good} 

Our goal is to incentivize players to follow \emph{good behavior}---truthfully
reporting their demand, and then faithfully following the suggested action of
the mediator. Formally, the good behavior strategy for player $i$ is $\xi_i(s_i) =
(s_i, \text{id})$ where $s_i$ is $i$'s actual demand, and $\text{id}\colon \cF
\to\cF$ is the identity map.  We write $\xi_i = \xi_i(s_i)$ for the good behavior strategy.

To accomplish this goal, we will design a mediator that is
``insensitive'' to the reported demand of each player.  Informally, if a
player's reported demand does not substantially effect the tolls chosen
by the mediator, or the paths suggested to \emph{other} players, then
a player has little incentive to lie about his demand (of course any
mediator with this property must necessarily allow the path suggested
to agent $i$ to depend strongly on agent $i$'s own reported demand!). We capture this notion of insensitivity using \emph{joint differential privacy}~\cite{KPRU14}, which is defined as follows.
\begin{definition}
  (Joint Differential Privacy \cite{KPRU14}) A randomized algorithm
  $\cM: \cS^n \to \mathcal{O}^n$, where $\mathcal{O}$ is an arbitrary output set for each player, satisfies
  $(\eps,\delta)$-joint differential privacy if for
  every player $i$, every pair $s_i,s_i' \in \cS$, any
  tuple $s_{-i} \in \cS^{n-1}$ and any $B_{-i}
  \subseteq \mathcal{O}^{n-1}$, we have
$
\Prob{}{\cM(s_i,s_{-i})_{-i} \in B_{-i}} \leq e^{\eps} \cdot \Prob{}{\cM(s_i', s_{-i})_{-i} \in B_{-i}} + \delta.
$
\end{definition}

Joint differential privacy (JDP) is a relaxation of the notion of \emph{differential privacy} (DP)~\cite{DMNS06}.  We state the definition of DP below, both for comparison, and because it will be important technically in designing our mediator.
\begin{definition}
  (Differential Privacy \cite{DMNS06}) A randomized algorithm $\cM:
  \mathcal{S}^n \to \mathcal{O}$ satisfies
  $(\eps,\delta)$-differential privacy if for any player $i$, any
  two $s_i, s_i' \in \cS$, any tuple
  $s_{-i} \in \mathcal{S}^{n-1}$, and any $B \subseteq
  \mathcal{O}$ we have $\Prob{}{\cM(s_i,s_{-i}) \in B} \leq e^{\eps} \cdot \Prob{}{\cM(s_i',
  s_{-i}) \in B} + \delta$.
\end{definition}

Note that JDP is weaker than DP, because JDP assumes that the output space of the algorithm is partitioned among the $n$ players, and the output to player $i$ can depend arbitrarily on the input of player $i$, and only the output to players $j \neq i$ must be insensitive to the input of player $i$.  This distinction is crucial in mechanism design settings---the output to player $i$ is a suggested route for player $i$ to follow, and thus should satisfy player $i$'s reported demand, which is highly sensitive to the input of player $i$.  Also note that since our mediator will output the same tolls to every player, the tolls computed by the mediator must satisfy standard DP.

A key property we use is that a JDP mediator that also computes an equilibrium of the
underlying game gives rise to an approximately truthful mechanism.  This
result was first shown in \cite{KPRU14, RR14}, although for simpler models that do not include tolls.
We now state and prove a simple extension of this result that is appropriate for our setting.
\begin{theorem}
\label{thm:privacyIncentives}
 Given routing game $\Gamma = (G,\ell,\bs)$ and upper bound $U$ on the tolls, let $M:(\bot \cup \cS)^n \to \cF^n\times[0,U]^m$ where $M(\bs') = \left(M_i^\cF(\bs'),M^\tau(\bs')\right)_{i\in [n]}$ satisfies
 \begin{enumerate}
  \item $M$ is $(\eps,\delta)$-joint differentially private.
  \item For any input demand profile $\bs$, we have with probability $1-\beta$ that $\bbx = \left(M_i^\cF(\bs)\right)_{i=1}^n$ is an $\eta_{eq}$-approximate Nash flow in the modified routing game $\Gamma^\tau = (G,\ell^M,\bs)$ where
  $$
  \ell_e^M(y) := \ell_e(y) + M^\tau_e(\bs) \quad \forall e \in E.
  $$
 \end{enumerate}
 Then the good
 behavior strategy $\xi = (\xi_1, \ldots, \xi_n)$ forms an
 $\eta$-approximate ex-post Nash equilibrium in $\Gamma_M =
 (G,\ell,\bs,M)$, where 
\ifnum\final=1 $ \else $$ \fi
 \eta = \eta_{eq} + m(U+n)(2\eps + \beta + \delta).
\ifnum\final=1 $ \else $$ \fi
 \label{thm:motivate}
\end{theorem}
\begin{proof}
We fix $\bs \in \cS^n$ to be the true source destination of the players. We consider a unilateral deviation $\xi_i'(s_i) = (s_i',f_i')$ for player $i$ to report $s_i'$ and use $f_i'$, which we write as $\xi_i'$.  We write the modified cost function for player $i$ in $\Gamma^{\tau}$ with tolls $\tau_e = M^\tau_e(\bs)$ to be
$$
\phi_i^\tau(x_i,\bbx_{-i}) = \sum_{e\in E} x_{i,e} \left(\ell_e\left(\sum_{j = 1}^n x_{j,e} \right) +\tau_e\right)
$$
We define the best response flow that player $i$ of demand $s_i$ can route given the flows of the other players to be
\ifnum\final=1 $ \else $$ \fi
BR_i^\tau(\bbx_{-i}) = \myargmin_{x_i \in \cF(s_i)}\left\{ \phi^\tau_i(x_i,\bbx_{-i}) \right\}.
\ifnum\final=1 $ \else $$ \fi

We first condition on the event that $M$ gives an $\eta_{eq}$-approximate Nash flow in $\Gamma^{\tau}$.
\begin{align*}
 \phi^M(s_i,(\xi_i,\xi_{-i}) ) = \Exp{(\bbx, \tau) \sim M(\bs) }{\phi_i^\tau(x_i,\bbx_{-i})} \leq \Exp{(\bbx, \tau) \sim M(\bs) }{\phi_i^\tau(BR_i^\tau(\bbx_{-i}),\bbx_{-i})}+ \eta_{eq}
\end{align*}
We then use the fact that $M$ is JDP.  We write $\bs' = (s_i',\bs_{-i})$.
\begin{align*}
 \phi^M(s_i,(\xi_i,\xi_{-i}) )& \leq e^\eps \left( \Exp{(\bbx,\tau)\sim M(\bs')}{\phi_i^\tau(BR_i^\tau(\bbx_{-i}),\bbx_{-i})}\right) + m(U+n)\delta +\eta_{eq}\\
 & \leq \Exp{(\bbx,\tau)\sim M(\bs')}{\phi_i^\tau(BR_i^\tau(\bbx_{-i}),\bbx_{-i})} + m(U+n)\left(2\eps+\delta \right)+\eta_{eq} \\
 & \leq \Exp{(\bbx,\tau)\sim M(\bs')}{\phi_i^\tau(f_i'(x_i),\bbx_{-i})} + m(U+n)\left(2\eps+\delta \right)+\eta_{eq}
\end{align*}
The first inequality comes from using the fact that $M$ is $(\eps,\delta)$-JDP and the fact that $\phi_i^\tau(\bbx) \leq m(U+n)$.  The second inequality uses the fact that $e^\eps \leq 1+ 2\eps$ for $\eps<1$.  The last inequality follows from the  fact that player $i$ can only do worse by not best responding to the other players' flows.  Lastly, we know that $M$ does not produce an $\eta_{eq}$-approximate Nash flow in $\Gamma^{\tau}$ with probability less than $\beta$, which gives the additional $\beta$ term in the theorem statement.
\end{proof}

The rest of the paper will be dedicated to constructing such a mediator that satisfies the hypotheses in Theorem \ref{thm:motivate}.  We now state the main result of our paper.
\begin{theorem}
 For routing games $\Gamma$ that satisfy Assumption \ref{assumption} and parameter $\beta>0$, there exists a mediator $M: \{\bot \cup \cS\}^n \to \cF^n \times [0,n\gamma]^m$ such that with probability $1-\beta$ good behavior forms an $\eta$-approximate ex-post Nash equilibrium in $\Gamma_M$ where
 $$
 \eta = \tilde O\left(m^{3/2} n^{4/5} \right)
 $$
 and the resulting flow from the good behavior strategy is $\eta_{opt}$-approximately optimal for
 $$
 \eta_{opt}=\tilde O\left(m n^{4/5} \right).
 $$
 \label{thm:main_one}
\end{theorem}


\jnote{I feel like we should restate the main theorem here now that all the notation is in place to state it formally.}
\rynote{Added.}

\section{Flow Mediator with Tolls}
We start by presenting a high level
overview of the design of our algorithm. Our goal is to design a mediator
that takes as input the demands, or source-destination pairs, $\bs$ of the players and
outputs a nearly optimal flow $\bbx\bl$ for $\Gamma =
(G,\ell,\bs)$ together with edge tolls $\tau$, such that the tolls are not
heavily influenced by any single player's report and no one's report has major influence on the flow induced by
the other players. Further, we need the tolls $\tau$ to be carefully computed
so that $\bbx\bl$ is also an approximate Nash flow in the 
instance $\Gamma^{\tau} = (G,\ell^\tau,\bs)$. We construct such a mediator
in the following way:
\begin{enumerate}
\item We compute an approximately optimal flow $\bbx\bl$ subject to
  JDP, using a privacy preserving variant of
  projected gradient descent. This ends up being the most technical
  part of the paper and so we leave the details to Section
  \ref{opt_jdp} and give the formal algorithm $\pigd$ \ifnum\appen=0 in the full version.\footnote{See \url{http://arxiv.org/abs/1502.04019} for the full version} \else in Algorithm
  \ref{PIGD}.\fi  For the rest of this section we assume we have
  $\bbx\bl$.
\item Given $\bbx\bl$, we need to compute the necessary tolls
  $\hat\tau$ such that players are approximately best responding in
  $\Gamma^{\hat\tau} = (G,\ell^{\hat\tau},\bs)$ when playing
  $\bbx\bl$.  We compute $\hat\tau$ as a function of a noisy version
  of the edge congestion $\hat\bby$ induced by the flow $\bbx\bl$ so
  that $\hat\tau$ is DP.  We give the procedure
  $\pcon$ that computes $\hat\bby$ in Algorithm \ref{PCON}. We must be cautious at this step because $\bbx\bl$ is only approximately optimal (and the
  tolls are computed with respect to a perturbed version of the
  induced congestion), so there may be a few players that are not playing
  approximate best responses in $\Gamma^{\hat\tau}$. We call these
  players \emph{unsatisfied}.
\item We show that the number of \emph{unsatisfied} players in
  $\Gamma^{\hat\tau}$ with flow $\bbx\bl$ is small, so we can modify
  $\bbx\bl$ by having the unsatisfied players play best responses to
  the induced flow. Because the number of unsatisfied players was
  small, we can show that this modification does not substantially reduce the payoff of the
  other players.  Therefore, if those players were playing approximate best responses before the modification, they will continue to do so after.  The procedure $\sbr$, given in Algorithm \ref{sbr}, ensures every player is approximately best
  responding. The result is a
  slightly modified flow $\hat\bbx$ which is nearly optimal in
  $\Gamma$ and an approximate Nash flow in $\Gamma^{\hat\tau}$.
\item The final output is then $\hat\bbx$ and $\hat\tau$.
\end{enumerate}

Our mediator $\med$ is formally given in Algorithm \ref{complete} and
is composed of the subroutines described above.  In $\med$ we
are using $\pigd$ as a black box that computes an $\alpha$-approximate
optimal flow\ifnum\appen=0 , which we show in the full version that we can set\else .  Theorem \ref{thm:approx_opt_flow} shows that we
can set\fi
\begin{equation}
\alpha = \tilde O \left( \frac{ \sqrt{n}m^{5/4}}{\sqrt{\eps}} \right). 
\label{eq:alpha}\end{equation}
  The rest of this paper is dedicated to analyzing the subroutines of $\med$.

\ifnum\final=1
 \begin{algorithm}[h]
 \SetAlgoNoLine
 $\med(\Gamma,\eps,\delta,\beta)$\;
 \KwIn{A routing game instance $\Gamma = (G, \ell, \bs)$; privacy parameter $(\eps,
     \delta)$; failure probability $\beta$}
 \KwOut{$\hat x_i$, a $(s_i^1, s_i^2)$-flow for each player $i\in [n]$, and a toll $\hat\tau_e$ for each edge $e\in E$}
 {
 \begin{enumerate}
     \item Compute an $\alpha$-approximately optimal flow $$\bbx\bl\gets\pigd\left(\Gamma, \frac{\eps}{4}, \frac{\delta}{2}, \frac{\beta}{2}\right).$$
     \item Compute congestion $\hat\bby \gets \pcon(\bbx\bl,\eps/4)$ and tolls $\hat \tau \gets \tau^*(\hat y_e)$ where $\tau^*(\cdot)$ is given in \eqref{eq:marginal-cost}.
     \item Improve some players' paths $$\hat \bbx \gets \sbr\left(\Gamma^{\hat\tau}, \hat\bby,\bbx\bl, 4\sqrt{mn\gamma \alpha} + \frac{8 \gamma m^2\log(m/\beta)}{\eps}\right).$$
 \end{enumerate}
 \Return $\hat \bbx$ and $\hat \tau$
 }
\caption{Flow Mediator with Tolls}\label{complete}
\end{algorithm}
\fi

\ifnum\final=0
\begin{algorithm}[h]
  \caption{Flow Mediator with Tolls}\label{complete}
  \begin{algorithmic}
    \INPUT A routing game instance $\Gamma = (G, \ell, \bs)$; privacy parameter $(\eps,
    \delta)$; failure probability $\beta$
    \OUTPUT $\hat x_i$, a
    $(s_i^1, s_i^2)$-flow for each player $i\in [n]$, and a toll $\hat
    \tau_e$ for each edge $e\in E$

    \Procedure{\med}{$\Gamma, \eps,\delta, \beta$}
    \begin{enumerate}
    \item Compute an $\alpha$-approximately optimal flow $$\bbx\bl\gets\pigd\left(\Gamma, \frac{\eps}{4}, \frac{\delta}{2}, \frac{\beta}{2}\right).$$
    \item Compute congestion $\hat\bby \gets \pcon(\bbx\bl,\eps/4)$ and tolls $\hat \tau \gets \tau^*(\hat y_e)$ where $\tau^*(\cdot)$ is given in \eqref{eq:marginal-cost}.
    \item Improve some players' paths $$\hat \bbx \gets \sbr\left(\Gamma^{\hat\tau}, \hat\bby,\bbx\bl, 4\sqrt{m\gamma \alpha} + \frac{8 \gamma m^2\log(2m/\beta)}{\eps}\right).$$
    \end{enumerate}
    \Return $\hat \bbx$ and $\hat \tau$
    \EndProcedure
  \end{algorithmic}
\end{algorithm}
\fi

\begin{remark}
  Throughout our discussion of the subroutines, we will sometimes say
  ``player $i$ plays...'' or ``player $i$ best responds to...'' to describe player $i$'s action in some flow computed by these subroutines. While
  these descriptions are natural, they could be slightly
  misleading. We want to clarify that our mediator mechanism is not
  interactive or online, and all the computation is done by the
  algorithm. The players will simply submit their private
  source-destination pairs and will only receive a suggested feasible
  path along with the tolls over the edges.
\end{remark}

\subsection{Private Tolls Mechanism}\label{phase2}


We show in this section that given an approximately optimal flow
$\bbx\bl$ we can compute the necessary tolls $\hat\tau$ in a
DP way.  Ultimately, we want to compute
\emph{constant} tolls, but a useful intermediate step is to consider
the following \emph{functional} tolls, which are edge tolls that can
depend on the congestion on that edge.  Specifically, we define the
\emph{marginal-cost toll} $\tau^*_e \colon \R\to\R$ for each edge $e \in E$
to be
\begin{equation}\label{eq:marginal-cost}
  \tau_e^*(y) = (y-1)(\ell_e(y) - \ell_e(y-1)),
\end{equation}
\jnote{Looking at this, you'd think there is a problem with $y \in
  (0,1)$.  Do we want to make it clear that the congestion will be a
  natural number?}\snote{although we will be applying this function to
  the noisy congestion later. }
 which gives rise to a different
routing game $\Gamma^{\tau^*} = (G, \ell^{\tau^*}, \bs)$ with latency
function defined as $\ell^{\tau^*}_e(y) = \ell_e(y) + \tau^*_e(y)$ for
 $e\in E$.

We first show that a marginal-cost toll enforces the optimal flow in
an atomic, unsplittable routing game, and then show how to use this
fact to privately compute \emph{constant} tolls that approximately
enforce the optimal flow at equilibrium. Recall the classical
potential function method \cite{MS96} for congestion games that
defines a potential function $\Psi: \R^{n \times m} \to \R$ such that
a flow $\bbx$ that minimizes $\Psi$ is also a (exact) Nash flow in
$\Gamma^{\tau^*} = (G,\ell^{\tau^*},\bs)$, where
\begin{equation}
\Psi(\bbx) :=\sum_{e\in E}\sum_{i=1}^{y_e} \ell^{\tau^*}_e(i) =\sum_{e\in
  E}\sum_{i=1}^{y_e} \left[\ell_e(i) + \tau^*_e(i)\right], \quad
\text{and } y_e = \sum_{i \in [n]} x_{i,e}.
\label{eq:potential}
\end{equation}

\begin{lemma}
  Let $\bbx^*$ be the (exact) optimal flow in routing game $\Gamma
  = (G,\ell,\bs)$, then $\bbx^*$ is a Nash flow in $\Gamma^{\tau^*} =
  (G,\ell^{\tau^*},\bs)$
  \label{lem:opt_eq}
\end{lemma}
\begin{proof}
  First, we show that $n\cdot \phi(\bbx) = \Psi(\bbx)$ where $\phi$ is given
  in \eqref{eq:min_cost_flow}:
\begin{align*}
  \Psi&(\bbx) = \sum_e \sum_{i = 1}^{y_e} \left[ \ell_e(i) +
    \tau^*_e(i)\right] = \sum_e \sum_{i=1}^{y_e} \left[\ell_e(i) + (i - 1)(\ell_e(i)
    -
    \ell_e(i - 1))\right] \\
  &= \sum_e \sum_{i = 1}^{y_e} \left[i\,\ell_e(i) -
    (i - 1)\,\ell_e(i - 1) \right] = \sum_e \left[ y_e\,\ell_e(y_e) - 0\,\ell_e(0)\right] =  \sum_e
  y_e\, \ell_e(y_e) = n\cdot\phi(\bbx).
\end{align*}

 Note that $\bbx^*$ minimizes the potential function $\Psi$.  We know
 from \cite{MS96} that the flow that minimizes the potential function
 $\Psi$ is a Nash flow of the routing game $\Gamma^{\tau^*}$.  Hence
 $\bbx^*$ is a Nash flow.
\end{proof}

Since we only have access to an approximately optimal flow $\bbx\bl$,
we will compute the marginal-cost tolls based on $\bbx\bl$ instead.
In order to release DP tolls, we compute them
using a private version $\hat y_e$ of the total edge congestion
$y_e=\sum_i x\bl_{i,e}$ that is output by $\pcon$ (presented in
Algorithm~\ref{PCON}). Using a standard technique in differential
privacy, we can release a private version of the edge congestion by
perturbing the congestion on each edge with noise from an
appropriately scaled Laplace distribution.  Since the analysis is
standard, we defer the details \ifnum\appen=0 to the full version\else to
Section~\ref{laplace}\fi. Lastly, to get the constant tolls for the
mediator $\med$, we will evaluate the marginal-cost toll function
on the perturbed edge congestion $\hat \bby$: set $\hat\tau_e =
\tau^*_e(\hat y_e)$ for $e\in E$.

\ifnum\final=1
 \begin{algorithm}[h]
 \SetAlgoNoLine
$\pcon(\bbx,\epsilon)$\;
    \KwIn{Flow $\bbx$, privacy parameter $\eps$}
    \KwOut{Aggregate flow $\hat \bby = (\hat y_e)_{e \in E}$}
    {\bf for } each edge $e\in E$ {\bf do} 
    $\quad $  write $\hat y_e = \sum_i x_{i,e}+ Z_e, \mbox{ where }Z_e {\sim}\Lap(m/\eps).$ \\
     $\qquad $ {\bf if } $\hat y_e > n$ {\bf then } $\hat y_e \gets n$.\\
    \Return $\hat \bby$
\caption{Private Congestion}\label{PCON}
\end{algorithm}
\fi

\ifnum\final=0
 \begin{algorithm}[h]
  \caption{Private Congestion}\label{PCON}
  \begin{algorithmic}[0]
    \INPUT Flow $\bbx$, privacy parameter $\eps$
    \OUTPUT Aggregate flow $\hat \bby = (\hat y_e)_{e \in E}$
    \Procedure{\pcon}{$\bbx, \eps$}
    \For{each edge $e\in E$}
    \State Let $\hat y_e = \sum_i x_{i,e}+ Z_e, \mbox{ where }Z_e 
           {\sim}\Lap(m/\eps).$
           \If{$\hat y_e > n$}
           \State $\hat y_e \gets n$.
           \EndIf
    \EndFor
    \Return $\hat \bby$
    \EndProcedure
  \end{algorithmic}
\end{algorithm}
\fi

To show that the constant tolls $\hat \tau$ are private, we
need to first show that the noisy congestion $\hat \bby$ output by
$\pcon$ is DP in the demands $\bs$. We will show later that $\pigd$ which computes $\bbx\bl$ is JDP in $\bs$.  We then use $\bbx\bl$ as input to $\pcon$, which we know is
DP with respect to any flow input
$\bbx$. To bridge the two privacy guarantees, we rely on the
following composition lemma (with proof in \ifnum\appen=0 the full version\else Appendix~\ref{sec:comp}\fi) to show that $\hat \bby$ is DP in $\bs$.

\begin{lemma}
\label{lem:djcomp}
Let $M_J:\cS^n \to \cX^n$ be $(\eps_J,\delta)$-jointly differentially
private.  Further, let $M_D:\cX^n \to O$ be
$\eps_D$-differentially private. If $M:\cS^n \to O$ is
defined as
\ifnum \final=1 $ \else $$ \fi
M(\bs) = M_D(M_J(\bs))
\ifnum \final=1 $ \else $$ \fi
then $M$ is $(2\eps_D+\eps_J,\delta)$-differentially private.
\end{lemma}

Now we are ready to establish the privacy guarantee of both $\hat \bby$
and $\hat \tau$.

\begin{corollary}
Given the approximately optimal flow $\bbx\bl$ computed from
$\pigd(\Gamma, \eps/4, \delta/2, \beta/2)$, the perturbed congestion
$\hat \bby$ output by $\pcon(\bbx\bl, \eps/4)$ and the constant tolls
$\hat \tau = (\tau^*_e(\hat y_e))_{e \in E}$ are $(3\eps/4,
\delta/2)$-differentially private in the demands $\bs$.
\label{cor:pcon_dp}
\end{corollary}
\begin{proof}
Note that $\bbx\bl$ is output by
$\pigd(\Gamma,\eps/4,\delta/2,\beta/2)$, so it is $(\eps/4,
\delta/2)$-JDP in $\bs$.  Using analysis of the Laplace mechanism\ifnum\appen=1 (Section~\ref{laplace})\fi, we know that
$\pcon(\bbx\bl, \eps/4)$ is $(\eps/4)$-DP in
$\bbx\bl$.  Therefore, the noisy congestion $\hat \bby$ output by the
composition of these two functions is $(3\eps/4,
\delta/2)$-DP by Lemma~\ref{lem:djcomp}. Since $\hat \tau$ is simply a post-processing of the noisy congestion
$\hat \bby$, we know that $\hat \tau$ is
$(3\eps/4,\delta/2)$-DP\ifnum\appen=1 by Lemma~\ref{lem:post}\fi.
\end{proof}

\subsection{Simultaneous Best Responses of Unsatisfied Players}

At this point of the mechanism, we have computed the approximately
optimal flow $\bbx\bl$ and constant tolls $\hat \tau$ that define the
tolled routing game $\Gamma^{\hat \tau}$. In this section, we show
how to modify $\bbx\bl$ to obtain a new approximately optimal flow $\hat \bbx$ that is also an approximate Nash equilibrium in the presence of the same
constant tolls $\hat \tau$.

Recall from Lemma~\ref{lem:opt_eq} that there is an exactly optimal flow $\bbx^*$ and functional tolls $\tau^*$ such that $\bbx^*$ is an exact Nash flow of the routing game under tolls $\tau^*$.  Our flow-toll pair $(\bbx\bl, \hat\tau)$ differs from $(\bbx^*, \tau^*)$ in three ways.
\begin{enumerate}
\item The flow $\bbx\bl$ is only \emph{approximately} optimal.
\item The tolls $\hat \tau$ we impose on the edges are \emph{constants} while
  the functional tolls $\tau^*$ may be functions of the congestion.
\item Tolls $\hat \tau$ are derived from noisy congestion $\hat \bby$, not the exact congestion $\bby\bl = \sum_{i} \bbx_i \bl$.
\end{enumerate}

As a result, there may be some \emph{unsatisfied players} who could significantly benefit from deviating from $\bbx\bl$.  We obtain the new approximate Nash flow $\hat \bbx$ by rerouting the unsatisfied players in $\bbx\bl$ along their best response route in the flow $\bbx\bl$ with constant edge tolls $\hat\tau$.  To analyze the new flow $\hat \bbx$, we show that there are not too many unsatisfied players.  Thus, even if we modify the routes of all of the unsatisfied players, the overall congestion does not change too much, and thus the players who were previously satisfied remain satisfied.

To determine if a player is unsatisfied and what their best response is, we need to know the costs they face for different paths, which depends on the flow
$\bby\bl = \sum_i \bbx_i\bl$.  However, to ensure privacy, we
only have access to a perturbed flow $\hat \bby$.  Thus, we will define unsatisfied players relative to this noisy flow $\hat \bby$ computed by $\pcon$.  More generally we can define the best response function of a player relative to any flow $\bby$.

Given any congestion $\bby$ (not necessarily even a sum of feasible individual flows) and routing game $\Gamma = (G, \ell, \bs)$, we define $c_{\bbx_i}(\bby)$
to be player $i$'s cost for routing on path $\bbx_i$ under the
congestion of $\bby$, that is
\begin{equation}
c_{\bbx_i}(\bby) = \sum_{e \in E} x_{i,e}\cdot \ell_e(y_e).
\label{eq:aggregate_cost}
\end{equation}
Note that $\sum_{i=1}^n c_{\bbx_i}(\bby) = n\phi(\bbx)$ and
$c_{\bbx_i}(\bby) = \phi_i(\bbx) $ when $y_e = \sum_{i = 1}^n x_{i,e}$
for $e \in E$. We then define the condition for being
unsatisfied with respect to congestion $\bby$ as follows.
\begin{definition}
  Given congestion $\bby$ and routing game $\Gamma = (G, \ell, \bs)$, we say that a
  player $i$ with $s_i$-flow $\bbx_i$ is
  \emph{$\rho$-unsatisfied with respect to $\bby$} if he could
  decrease his cost by at least $\rho$ via a unilateral deviation.  That is, there exists a path $\bbx_i' \in \cF(s_i)$ such that
\ifnum \final=1 $ \else $$ \fi
c_{\bbx_i '}(\bby') \leq c_{\bbx_i }(\bby) -\rho
\ifnum \final=1 $ \else $$ \fi
where $\bby' = \bby - \bbx_i + \bbx'_i$ is the flow that would result from player $i$ making this deviation.
If player $i$ is not $\rho$-unsatisfied, then we say $i$ is \emph{$\rho$-satisfied}.  We will sometimes omit $\bby$ if it is clear from context.
\end{definition}

The next lemma bounds the number of unsatisfied players in $\bbx\bl$ in the routing game $\Gamma^{\hat \tau} = (G,\ell + \hat\tau,\bs)$ with respect to the noisy congestion $\hat \bby$.

\begin{lemma}\label{lem:unsatisfied-final}
  Let $\bbx\bl$ be an $\alpha$-approximately optimal flow, $\hat\bby =
  \pcon(\bbx\bl, \eps)$ be the noisy aggregate flow, and $\hat \tau =
  \tau^*(\hat\bby)$ be a vector of constant tolls.  Then with
  probability at least $1-\beta$ for $\beta>0$, there are at most $\sqrt{n\alpha / 4
    m\gamma}$ players who are $\hat\zeta_\eps$-unsatisfied players in
  $\Gamma^{\hat\tau}$ with respect to the congestion $\hat\bby$, for
    \begin{equation}
  \hat \zeta_\eps = 4\sqrt{mn\gamma \alpha} +  8 \frac{\gamma m^2\log(m/\beta)}{\eps}.
  \label{eq:zeta}
  \end{equation}

\end{lemma}
We will now give a rough sketch of the proof, \ifnum\appen=0 formally given in the full version\else.  The full proof appears in Appendix~\ref{append_sbr}.\fi
\ifnum\final=0
\begin{proof}[Proof Sketch]
\else
\begin{proof}[Sketch]
\fi

First, we will consider the routing game $\Gamma^{\tau^*}$ under the (functional) marginal-cost toll.  We will also assume for now that we have the exact congestion $\bby\bl = \sum_i \bbx_i\bl$. Recall from Lemma~\ref{lem:opt_eq} that the potential function $\Psi$ for this game is equal to the total congestion cost $n \cdot \phi$. \jnote{``Coincides'' is a pretty vague word.  Can we be more precise?}\snote{changed to ``equals to''} Since $\bbx\bl$ is an $\alpha$-approximate optimal flow, it also approximately minimizes $\Psi$ up to error $n\cdot \alpha$.  The construction of $\Psi$ is such that if a player who is $\rho$-unsatisfied with respect to $\bby\bl$ plays her best response, then $\Psi$ decreases by at least $\rho$.  Therefore the number of $\rho$-unsatisfied players with respect to $\bby\bl$ is at most $n \alpha/\rho$.  Here we are intentionally being slightly imprecise to ease exposition.  See the full proof for details.

Now, consider the routing game $\Gamma^{\tau} = (G, \ell + \tau, \bs)$ that arises from using the constant tolls $\tau = \tau^*(\bby\bl)$.  Note that under functional tolls $\tau^*$, when a player best responds, the tolls may change, however under constant tolls $\tau$ the tolls do not change.  This might increase the number of players who can gain by deviating.  However, notice that when one player changes their route, the tools $\tau^*_e$ and $\tau_e$ can only change by $\gamma$, since $\tau^*_e$ is $\gamma$-Lipschitz.  Thus changing from tolls $\tau^*$ to $\tau$ can only change the cost any player faces on any route by $m\gamma$.  Therefore, we can argue that the number of $(\rho + 2m\gamma)$-unsatisfied players with respect to $\bby\bl$ in the game $\Gamma^{\tau}$ is also at most $n\alpha / \rho$.

The last issue to address is that we compute the tolls from the noisy congestion $\hat\bby$ instead of the exact congestion $\bby\bl$.  This has two effects: 1) the constant tolls $\hat \tau = \tau^*(\hat\bby)$ are different from the constant tolls $\tau = \tau^*(\bby\bl)$ analyzed above and 2) we want to measure the number of unsatisfied players with respect to $\hat\bby$ instead of $\bby\bl$.  We can address both of these issues using the fact that the noise is small on every edge.  Therefore $|y_e - \hat y_e|$ is small, and since $\tau^*_e$ is Lipschitz, $|\tau_e - \hat\tau_e|$ is small as well.  In the full proof we carefully account for the magnitude of the noise and its effect on the cost faced by each player to obtain the guarantees stated in the lemma.
\end{proof}

We have so far shown that there might be a few players that are
unsatisfied with their current route in $\Gamma^{\hat\tau} =
(G,\ell+\hat \tau ,\bs)$ when they only know a perturbed version of
the congestion $\hat \bby$. We then let these unsatisfied players
simultaneously change routes to the routes with the lowest cost
(according to the cost $c_{\bbx_i}(\bby)$).  This procedure, $\sbr$,
is detailed in Algorithm~\ref{sbr}.

\ifnum\final=1
\begin{algorithm}[h]
 \SetAlgoNoLine
  $\sbr(\Gamma, \bby, \bbx, \zeta)$\;
   \KwIn{Routing game instance $\Gamma$, congestion $\bby$, flow assignment $\bbx$, satisfaction parameter $\zeta$}
   \KwOut{New flow assignment $\hat \bbx$}
  {
      Let $\hat \bbx \gets \bbx$ \\
    {\bf for } each player $i \in [n]$ {\bf do }\\
     $\qquad $ {\bf if } $i$ with flow $\bbx_i$ is $\zeta$-unsatisfied with respect to congestion $\bby$ in game $\Gamma$ \\
     $\qquad\qquad $ Replace $\hat \bbx_i$ by the route with the lowest cost given congestion $\bby$.\\
     $$
    \hat \bbx_i \gets \argmin_{\bbx_i'}\left\{c_{\bbx_i'}\left(\bby' \right) \right\} \qquad \text{(breaking ties arbitrarily)}
     $$
    $\qquad\qquad $ Where $y_e' = y_e - 1$ if $x_{i,e} = 1, x_{i,e}' = 0; y_e' = y_e+1 $ if $ x_{i,e} = 0, x_{i,e}' = 1; $ else $ y_e = y_e'.$
  \Return $\hat \bbx$
  }
\caption{Private Best Responses}\label{sbr}
\end{algorithm}
\fi

\ifnum\final=0
\begin{algorithm}[h]
  \caption{Private Best Responses}\label{sbr}
  \begin{algorithmic}[0]
    \INPUT Routing game instance $\Gamma$, congestion $\bby$, flow assignment $\bbx$, satisfaction parameter $\zeta$
    \OUTPUT New flow assignment $\hat \bbx$
    \Procedure{\sbr}{$\Gamma, \bby, \bbx, \zeta$}
    \State Let $\hat \bbx \gets \bbx$
    \For{each player $i \in [n]$}
    \If{$i$ with flow $\bbx_i$ is $\zeta$-unsatisfied with respect to congestion $\bby$ in game $\Gamma$}
    \State Replace $\hat \bbx_i$ by the route with the lowest cost given congestion $\bby$.
    $$
    \hat \bbx_i \gets \argmax_{\bbx_i'}\left\{c_{\bbx_i'}\left(\bby' \right) \right\} \qquad \text{(breaking ties arbitrarily)}
    $$
\State Where $y_e' = y_e - 1$ if $x_{i,e} = 1, x_{i,e}' = 0; y_e' = y_e+1 $ if $ x_{i,e} = 0, x_{i,e}' = 1; $ else $ y_e = y_e'.$
    \EndIf
    \EndFor
    \Return $\hat \bbx$
    \EndProcedure
  \end{algorithmic}
\end{algorithm}
\fi

We are now ready to show that the final flow assignments $\hat \bbx$
resulting from the procedure $\sbr(\Gamma^{\hat\tau}, \bbx\bl,
\hat\zeta_\eps)$, where $\bbx\bl$ is an $\alpha$-approximate optimal flow in
$\Gamma$ and $\hat\zeta_\eps$ is given in \eqref{eq:zeta}, forms an
approximate Nash equilibrium in the game $\Gamma^{\hat\tau}$ and remains an
approximately optimal flow for the original routing game instance $\Gamma$.

\begin{lemma}
Fix any $\alpha > 0$ and $\beta,\eps\in (0,1)$. Let $\Gamma =
(G,\ell,\bs)$ be a routing game and $\bbx\bl$ be an
$\alpha$-approximately optimal flow in $\Gamma$.  Let $\hat\bbx =
\sbr(\Gamma^{\hat\tau},\hat\bby, \bbx\bl,\hat\zeta_\eps)$ for
$\hat\zeta_\eps$ given in \eqref{eq:zeta}, $\hat\bby =
\pcon(\bbx\bl,\eps)$, and $\hat\tau = \tau^*(\hat\bby)$. Then with
probability at least $1 - \beta$, $\hat\bbx$ is an
$\eta_{eq}(\alpha)$-Nash flow in $\Gamma^{\hat\tau} =
(G,\ell+\hat\tau,\bs)$ where
\begin{equation}
\eta_{eq}(\alpha) = O\left(\sqrt{mn\alpha } + \frac{m^2\log(m/\beta)}{\eps} \right).
\label{eq:eta_eq}
\end{equation}
and $\hat\bbx$ is an $\eta_{opt}(\alpha)$-approximate Nash flow in $\Gamma$ where
\begin{equation}
\eta_{opt}(\alpha) = O\left(\alpha + \sqrt{mn\alpha}  \right).
\label{eq:eta_opt}
\end{equation}
\label{lem:sbr_eq_opt}
\end{lemma}
\begin{proof}
First, to show that $\hat \bbx$ forms an approximate Nash flow, we
need to argue that all players are approximately satisfied with respect to the
actual congestion $\bby = \sum_i\hat\bbx_i$. As an intermediate step,
we will first show that all players in $\hat \bbx$ are approximately
satisfied with the input perturbed congestion $\hat\bby$.

By Lemma~\ref{lem:unsatisfied-final}, we know that the number of
$\hat\zeta_\eps$-unsatisfied players that deviate in our instantiation
of $\sbr$ is bounded by
\ifnum \final=1 $ \else $$ \fi
\sqrt{n\alpha}/(2\sqrt{m\gamma})\equiv K.
\ifnum \final=1 $ \else $$ \fi
After these players' joint deviation, the congestion on any path is
changed by at most $m\,K$, so the total cost on any path is changed by at
most $m\gamma K = \sqrt{m n\alpha\gamma}/2$. Therefore, the players that
deviate are $\sqrt{mn\alpha\gamma}$-satisfied in $\Gamma^{\hat\tau}$
with respect to congestion $\hat\bby$ after the simultaneous moves.  Similarly,
the players that were originally $\hat\zeta_\eps$-satisfied in
$\Gamma^{\hat\tau}$ with congestion $\hat\bby$ remain $(\hat\zeta_\eps
+ \sqrt{m n \alpha\gamma})$-satisfied with $\hat\bby$ even after the
joint deviations.

From standard bounds on the tails of Laplace distribution\ifnum\appen=1
(Lemma~\ref{lem:lap_utility})\fi, we can bound the difference between
$\hat\bby$ and $\sum_i \bbx\bl_i$: with probability at least
$1-\beta$,
\[
\|\hat \bby - \sum_i \bbx\bl_i \|_\infty \leq 2  m\log(m/\beta)/\eps
\]
Since the number of players that deviate in $\sbr$ is bounded by $K$,
we could bound $\|\bby - \sum_i\bbx\bl_i\|_\infty\leq K$. By triangle
inequality, we get
\[
\|\hat\bby - \bby\|_\infty \leq 2 m\log(m/\beta)/\eps + K.
\]
Since all players in $\hat\bbx$ are $(\hat\zeta_\eps +
\sqrt{mn \alpha\gamma})$-satisfied with congestion $\hat \bby$, \ifnum\appen=0 we show in the full version \else by
Lemma~\ref{lem:transfer}, we know\fi that they are also
$\eta_{eq}$-satisfied with the actual congestion $\bby$, where
\[
\eta_{eq} = \hat\zeta_\eps + \sqrt{mn\alpha\gamma} + \frac{4\gamma
  m^2\log(m/\beta)}{\eps} + 2K\gamma m = 6\sqrt{mn\alpha\gamma} +
\frac{12\gamma m^2 \log(m/\beta)}{\eps}.
\]

 Hence, the flow $\hat\bbx$ forms an $\eta_{eq}$-approximate Nash flow
 in game $\Gamma^{\hat\tau}$.  To bound the cost of $\hat \bbx$, note that for each edge $e$, the
number of players can increase by at most $K$. Let $\bby\bl = \sum_i
\bbx_i\bl$, then for each edge,  $y_e \ell_e( y_e) - y\bl_e
 \ell_e\left(y\bl_e\right) \leq n\gamma K + n\gamma = nK(\gamma+1).$

Therefore, the average cost for $\hat \bbx$ is
\begin{align*}
\phi(\hat\bbx) & = \frac{1}{n}\sum_{e\in E} y_e\ell_e( y_e) \leq \frac{1}{n}\sum_{e\in
  E} y_e\bl \ell_e(y_e\bl) + m K(\gamma+1) \leq \OPT(\bs) + \alpha +
\frac{\sqrt{mn\gamma\alpha}}{2}+ \frac{\sqrt{mn\alpha}}{2\sqrt{\gamma}}
\end{align*}
This completes the proof.
\end{proof}

\subsection{Analysis of $\med$}
Now that we have analyzed the subroutines $\pcon$ and $\sbr$ along
with computing the private tolls $\hat\tau$, we are ready to analyze
the complete mediator $\med$.  Note that in this analysis we will
assume that the subroutine $\pigd$ is a blackbox that is JDP and computes an approximately optimal flow in
$\Gamma$.

We first prove that the mediator $\med$ is JDP.
This will give the first condition we require of our mediator in Theorem
\ref{thm:motivate}.  A useful tool in proving mechanisms are JDP is the billboard lemma, which states at a high level
that if amechanism can be viewed as posting some public signal
(i.e. as if on ``a billboard'') that is DP in the
players' demands, from which (together with knowledge of their own demand)
players can derive their part of the output of the mechanism, then the
resulting mechanism is JDP.
\begin{lemma}[Billboard Lemma \cite{RR14, HHRRW14}] Let $\cM:
  \cS^n \to \mathcal{O}$ be an $(\eps,
  \delta)$-differentially private mechanism and consider any function
  $\theta: \cS \times \mathcal{O} \to \mathcal{A}$. Define
  the mechanism $M': \cS^n \to \mathcal{A}^n$ as follows: on
  input $\mathbf{s}$, $\cM'$ computes $o = \cM(\mathbf{s})$, and then
  $\cM'(\bs)$ outputs to each $i$:
$$\cM'(\bs)_i = \theta(s_i,o).$$
$\cM'$ is then $(\eps, \delta)$-jointly differentially private.
\label{lem:billboard}
\end{lemma}

We show that $\med$ is jointly differentially private via the billboard
lemma.
\begin{theorem}
For $\eps,\delta,\beta>0$, the procedure
$\med(\Gamma,\eps,\delta,\beta)$ in Algorithm \ref{complete} is
$(\eps,\delta)$-joint differentially private in the player's input
demands $\bs$.
\end{theorem}
\begin{proof}
In order to show JDP using the Billboard Lemma,
we need to show that for each player $i$, the output flow $\hat
\bbx_i$ and toll vector $\hat \tau$ can be computed only based on
$i$'s demands $s_i$ and some $(\eps, \delta)$-DP
signal.

In \ifnum\appen=0 the full version\else Theorem~\ref{thm:pigd_jdp}\fi, we show that the subroutine
$\pigd(\Gamma, \eps/4, \delta/2, \beta/2)$ operates in the Billboard
model, and can be computed from some $(\eps/4,
\delta/2)$-DP billboard signal $\Lambda$.


Note that the output flow $\hat \bbx_i$ for each player $i$ produced
by $\sbr(\Gamma^{\hat\tau},\hat\bby,\bbx\bl,\hat\zeta_{\eps/4})$ is
just a function of the perturbed congestion $\hat\bby$, $\bbx_i\bl$
and player $i$'s demand. Recall that we know that $\hat\bby =
\pcon(\bbx\bl,\eps/4)$ is $(3\eps/4, \delta/2)$-DP
in $\bs$ by Corollary \ref{cor:pcon_dp}. Therefore, the output flow
$\hat\bbx_i$ for each $i$ is just a function of the $(\eps,
\delta)$-DP signal $(\Lambda, \hat\bby)$, and $i$'s demand
$s_i$. Also, the tolls $\hat\tau$ are computed as a function only of
$\hat\bby$. Therefore, by the Billboard Lemma~\ref{lem:billboard}, the
mediator $\med(\Gamma, \eps,\delta, \beta)$ satisfies $(\eps,
\delta)$-JDP.
\end{proof}

Now we give the appropriate choices of the parameters $(\eps, \delta,
\beta)$ for $\med(\Gamma,\eps,\delta,\beta)$ that leads to our main
result in \ifnum\final=1 Theorem \ref{thm:main_one} \else the following theorem.  This result follows from
instantiating Theorem~\ref{thm:motivate} with a JDP algorithm that computes an approximately optimal flow
$\hat\bbx$ and tolls $\hat \tau$ such that $\hat \bbx$ forms an
approximate equilibrium in the routing game with tolls $\hat\tau$.\fi
 \begin{proof}[\ifnum\final=0 Proof \fi of Theorem \ref{thm:main_one}]
 Given any routing game instance $\Gamma=(G, \ell, \bs)$,  we first show that $\med$ is a
mediator that makes good behavior an $\eta$-approximate Nash
equilibrium of the mediated game $\Gamma_{\med}$ for $\eta = \tilde O\left(m^{3/2} n^{4/5} \right).$

We assume that $\pigd(\Gamma,\eps/2,\delta/2,\beta/2)$ produces
  an $\alpha$-approximate optimal flow $\bbx\bl$ with probability
  $1-\beta/2$ (leaving the formal proofs to \ifnum\appen=0 the full version\else Theorem \ref{thm:pigd_jdp}
  and Theorem \ref{thm:approx_opt_flow}\fi) where $\alpha$ is given  in
  \eqref{eq:alpha}. Consider the instantiation of
  $\med(\Gamma,\eps,\delta,\beta)$ with \ifnum\final=0\[\else $\fi\eps =  \frac{\sqrt{m}}{n^{1/5}}, \ifnum\final=0 \quad \fi \delta =  n^{-2}, \ifnum\final=0 \quad \fi \beta =  n^{-2}.\ifnum\final=0\]\else $\fi
  
 Given the functional tolls $\tau^*$ defined in~\eqref{eq:marginal-cost} and the fact that if we ever get an edge congestion $\hat y_e > n$ from the output of $\pcon$ then we round it down to $n$, so the edge tolls $\hat \tau_e$ are never bigger than $n\gamma$.  Using our bound for $\eta_{eq}(\alpha)$ in \eqref{eq:eta_eq} and setting $\eta_{eq} = \eta_{eq}(\alpha)$ where $\alpha$ is \ifnum\final=0 as above\else given in \eqref{eq:alpha}\fi, from Theorem \ref{thm:motivate} we have with probability $1-\beta$ the bound $\eta \leq \eta_{eq} + m(U + n)(2\eps + \beta + \delta)
  = \tilde O\left( m^{3/2} n^{4/5}  \right)$.

We then show that good behavior results in an $\eta_{opt}$-approximately optimal flow for the original routing game instance $\Gamma$, where \ifnum\final=0\[\else $\fi \eta_{opt} = \tilde O\left( mn^{4/5} \right).\ifnum\final=0\]\else $\fi  It then follows that $\eta_{opt} = \eta_{opt}(\alpha)$ from \eqref{eq:eta_opt} and for $\alpha$ given in \eqref{eq:alpha}.
\end{proof}

\section{Computing an Approximately Optimal Flow under JDP} \label{opt_jdp}

In this section we show how to compute an approximately optimal flow
$\bbx\bl$
under joint differential privacy. 
We first consider a convex relaxation of the problem of minimizing
social cost in the routing game instance $(\Gamma,\ell,\bs)$.  Let
$\cF^R(\bs)\subseteq [0,1]^{n \times m}$ be the set of feasible
\emph{fractional flows} (i.e.~the convex relaxation of the set
$\cF(\bs)$).  Then the optimal fractional flow is given by the convex
program:
\begin{align}
  \min & \qquad c(\bby) = \frac{1}{n}\sum_{e \in E} y_e \ell_e(y_e) \label{eq:convex_program}\\
  \text{such that} & \qquad \bbx \in \cF^R(\bs) \subseteq [0,1]^{n \times m} \nonumber\\
  & \qquad y_e =\sum_{i=1}^n x_{i,e} \qquad \forall e \in E, \quad \forall i \in [n] \label{eq:violation}
 \end{align}
Note that the second derivative of $y_e \ell_e(y)$ is $2\ell'_e(y_e) +
y_e \ell''_e(y_e)$.  Since $\ell_e$ is assumed to be convex and
nondecreasing, the second derivative is non-negative as long as $y_e
\geq 0$. Hence the objective function $c$ of this program is
indeed convex on the feasible region.  \jnote{Justifying that this
  program is convex is important!  Don't leave it in a footnote.}
\jnote{We hadn't been assuming $\ell$ is differentiable, so I added
  that.  We probably don't need it.}

We write $\cG^R(\bs) := \cF^R(\bs) \times [0,n]^m$ to denote the space
where the decision variables reside, i.e. $(\bbx,\bby) \in
\cG^R(\bs)$. Given any demands $\bs$, we write
$\OPT^{R}(\bs)$ to denote the optimal objective value of the convex
program and $\OPT(\bs)$ to be the optimal objective value when $\bbx \in \cF(\bs)$.  Note that we always have
$\OPT^{R}(\bs) \leq \OPT(\bs)$

Our goal is to first compute an approximately optimal solution to the
relaxed convex program, and then round the resulting fractional
solution to be integral. We then show that the final solution is an
approximately optimal flow to the original instance
$\Gamma$.


\subsection{The JDP Gradient Descent Algorithm} \label{sec:pgdalg}
We will work extensively with the Lagrangian of our problem.  For each
constraint of~\eqref{eq:violation}, we introduce a dual variable
$\lambda_e$. The Lagrangian is then
\ifnum \final=1 $ \else $$ \fi
\cL(\bbx, \bby, \lambda) = c(\bby) - \sum_{e\in E} \lambda_e \left( \sum_i
x_{i,e} - y_e \right).
\ifnum \final=1 $ \else $$ \fi
Since our convex program satisfies Slater's condition \cite{Slater59}, we know
that strong duality holds:

\begin{equation}
\max_{\lambda\in \mathbb{R}^m} \min_{(\bbx, \bby) \in \cG^R(\bs)} \cL(\bbx, \bby, \lambda) =
\min_{(\bbx, \bby)\in \cG^R(\bs)} \max_{\lambda\in \mathbb{R}^m} \cL(\bbx, \bby, \lambda) = \OPT^R(\bs).
 \label{eq:zero_sum}
\end{equation}

We will interpret the Lagrangian objective as the payoff function of a
zero-sum game between the minimization player, who plays flows $\bbz =
(\bbx, \bby)$, and the maximization player, who plays dual variables
$\lambda$. We will abuse notation and write $\cL(\bbz, \lambda) =
\cL(\bbx, \bby, \lambda)$.  We refer to the game defined by this
payoff matrix the \emph{Lagrangian game}.  We will privately compute
an approximate equilibrium of the Lagrangian game by simulating
repeated plays between the two players.  In each step, the dual player
will play an approximate best response to the flow player's
strategy. The flow player will update his flow using a no-regret
algorithm.

In particular, the flow player uses an online gradient descent algorithm to
produce a sequence of $T$ actions $\{\bbz^{(1)}, \ldots, \bbz^{(T)}\}$ based
on the loss functions given by the dual player's actions $\{\lambda^{(1)},
\ldots, \lambda^{(T)}\}$. At each round $t = 1, \ldots, T$, the flow player
will update both $\bbx^{(t)}$ and $\bby^{(t)}$ using the projected gradient
update step $\gd$ in Algorithm~\ref{GD}.
\ifnum\final=1
\begin{algorithm}[h]
 \SetAlgoNoLine
 $\gd(\cD, r, \omega,\eta)$\;
  \KwIn{Convex feasible domain $\cD$, a convex function $r$,
  some $\omega \in \cD$, and learning parameter $\eta$.} 
  \KwOut{Some new $\omega' \in\cD$. }
  We define the projection map $\Pi_{\cD}$ as $\Pi_{{\cD}}(v') = \argmin_{v \in \cD} || v - v'||_2$.\\
We then set $\omega' \gets \Pi_{\cD}(\omega - \eta \nabla r(\omega))$.\\
\Return $\omega'$
\caption{Gradient Descent with Projection} \label{GD}
\end{algorithm}
\fi

\ifnum\final=0
\jnote{We wrote ``define the projection map $\Pi_{\hat\cD}$.  The hat seemed like an error.}
\begin{algorithm}[h]
\caption{Gradient Descent with Projection} \label{GD}
\begin{algorithmic}[0]
  \INPUT Convex feasible domain $\cD$, a convex function $r$,
  some $\omega \in \cD$, and learning parameter $\eta$. \OUTPUT
  Some new $\omega' \in\cD$. \Procedure{\gd}{$\cD, r,
    \omega,\eta$} \State We define the projection map $\Pi_{\cD}$ as
$$
\Pi_{{\cD}}(v') = \argmin_{v \in \cD} || v - v'||_2
$$
\State We then set
$$
\omega' \gets \Pi_{\cD}(\omega - \eta \nabla r(\omega))
$$
\Return $\omega'$
\EndProcedure
\end{algorithmic}
\end{algorithm}
\fi

In order to reason about how quickly the projected gradient procedure
converges to an approximately optimal flow, we need to bound the
diameter of the space of dual solutions.  We will also need to argue
that bounding the space of feasible dual solutions does not affect the
value of the game.  Specifically, we will bound the dual players'
action to the set
\ifnum\final=1 $ \else \begin{equation} \fi
 \cB = \{\lambda\in \mathbb{R}^{m} \mid \|\lambda\|_1 \leq 2m\},
\ifnum\final=1 $ \else \label{eq:dual_domain} \end{equation} \fi
Then fixing a flow played by the primal player, the dual player's best
response is simply to select an edge $e$ where the
constraint~\eqref{eq:violation} is most violated and set
$\lambda^{(t)}_e = \pm 2m$.  Notice that, since the constraints depend
on the source/sink pairs, and we need to ensure joint differential
privacy with respect to this data, we must select the most violated
constraint in a way that maintains privacy.  Using a straightforward
application of the DP exponential mechanism
\cite{MT07}, we can obtain a constraint that is approximately the most
violated.  Since this step is standard, we defer the details to the
\ifnum\appen=0 full version\else appendix\fi.

From the repeated plays of the Lagrangian Game, we will obtain a
fractional solution $\overline {\bf z} = (\overline \bbx, \overline
\bby)$ to the convex program. Finally, we will round the fractional
flow $\overline\bbx$ to an integral solution $\bbx\bl$ for the
original minimum-cost flow instance $\Gamma = (G,\ell,\bs)$ using the
rounding procedure $\psrr$ proposed by~\cite{RT87}, given in \ifnum\appen=0 the full version\else Algorithm
\ref{PS}\fi.  The full procedure $\pigd$ is given in \ifnum\appen=0 the full version\else Algorithm
\ref{PIGD}\fi.

\ifnum\final=0
\begin{algorithm}[h]
\caption{Path Stripping and Randomized Rounding} \label{PS}
\begin{algorithmic}[0]
  \INPUT A fractional flow solution $\overline\bbx_i \in \cF^R(\bs_i)$ for player $i$
  \OUTPUT An integral flow solution $\bbx_i \in \cF(\bs_i)$ for player $i$
  \Procedure{\psrr}{$\overline\bbx_i$}
\State Let $\Lambda_i = \{P_j\}$ be the set of $(s_i^1, s_i^2)$-paths in $G$
\For{each path $P_j$}
\State Let $w_j = \min\{\overline x_{i, e}\mid e\in P_j\}$
\For{each edge $e\in P_j$}
\State Let $\overline x_{i, e} \gets x_{i,e} - w_j$
\EndFor
\EndFor
\State Sample a path $P$ from $\Lambda_i$ such that $\Prob{}{P = P_j} = w_j$
\For{each edge $e\in E$}
\State Let $x_{i, e} = \mathbb{I}[e\in P]$
\EndFor
\Return $\bbx_i$
\EndProcedure
\end{algorithmic}
\end{algorithm}
\fi

\ifnum\final=0

\ifnum\final=1
\section{Details for JDP Gradient Descent}
\fi

\ifnum\final=1
\begin{algorithm}[h]
 \SetAlgoNoLine
 $\psrr(\overline\bbx_i)$\;
  \KwIn{A fractional flow solution $\overline\bbx_i \in \cF^R(\bs_i)$ for player $i$.} 
  \KwOut{An integral flow solution $\bbx_i \in \cF(\bs_i)$ for player $i$. }
Let $\Lambda_i = \{P_j\}$ be the set of $(s_i^1, s_i^2)$-paths in $G$ \\
{\bf for } each path $P_j$ {\bf do } let $w_j = \min\{\overline x_{i, e}\mid e\in P_j\}$ \\
$\qquad${\bf for } each edge $e\in P_j$ {\bf do } let $\overline x_{i, e} \gets x_{i,e} - w_j$\\
Sample a path $P$ from $\Lambda_i$ such that $\Prob{}{P = P_j} = w_j$ \\
{\bf for } each edge $e\in E$ {\bf do} $ \quad $ $x_{i, e} = \mathbb{I}[e\in P]$\\
\Return $\bbx_i$
\caption{Path Stripping and Randomized Rounding} \label{PS}
\end{algorithm}

\fi

\ifnum\final=1
\begin{algorithm}[]
 \SetAlgoNoLine
 $\pigd(\Gamma,\epsilon,\delta, \beta)$\;
\KwIn{Routing Game $\Gamma = (G,\ell,\bs)$; privacy parameters $(\eps, \delta)$; failure probability $\beta$}
\KwOut{$\bbx_i\bl$, a $s_i = (s_i^1, s^2_i)$ flow for each player $i \in [n]$}
Define the following quantities:
$$
T \gets  \Theta\left( \frac{\epsilon n\sqrt{m}}{\log(mn/\beta) \sqrt{\log(1/\delta)}} \right)
\qquad \eps' \gets \eps/\sqrt{8T\ln(1/\delta)}
\qquad \eta_y \gets \frac{D_y}{G_y \sqrt{T}} \qquad \eta_x \gets \frac{D_x}{G_x \sqrt{T}}
$$
$$
\qquad G_y \gets \sqrt{ (m-1)(\gamma+1)^2 + (\gamma+1 + 2m)^2}\qquad D_y \gets n\sqrt{m} 
$$
$$
\qquad G_x \gets 2m\sqrt{n} \qquad D_x \gets\sqrt{mn}  
$$
Initialize: $\bby^{(1)} \in [0,n]^m$ and $\bbx^{(1)} \in
\cF^R(\bs)$. Let $\bbz^{(1)} \gets (\bbx^{(1)}, \bby^{(1)} )$ \\
Define the quality score $q : \cG(\bs) \times  \left((+, -) \times E \right) \to
\R$:
\[
f_e(\bbz) \gets \sum_i x^i_e - y_e \qquad q(\bbz,  (+, e)) \gets + f_e(\bbz) \qquad q(\bbz,  (-, e)) \gets -f_e(\bbz).
\]
{\bf for } $t = 1, \cdots, T$ {\bf do} \\
$\qquad$ Let $(\bullet^{(t)}, e^{(t)}) \gets \cM_{E}(\bs, q, \epsilon')$ (The Exponential Mechanism)\\
$\qquad$ Approximate best-response for the dual player $\lambda\utt$:\\
$\qquad$ {\bf if} $\bullet\utt = +\quad$ {\bf then } $\quad\lambda_{e^{(t)}}^{(t)} \gets -2m$ \\
$\qquad$ {\bf else} $\quad\lambda_{e^{(t)}}^{(t)} \gets + 2m$ \\
$\qquad$ {\bf for } $e' \in E\setminus\{e^{(t)}\}$ {\bf do } $\lambda_{e'}^{(t)} = 0$ and Gradient Descent update on the primal:\\
$\qquad \qquad$ Take a step to improve the individual flow variables $\bbx^{(t)}$:
\[\bbx^{(t+1)} \gets \gd(\cF(\bs),\cL (\cdot,\bby^{(t)},\lambda^{(t)}),\bbx^{(t)},\eta_x)\]
$\qquad \qquad$ Take a step to improve the congestion variables $\bby^{(t)}$:
\[ \bby^{(t+1)}\gets\gd([0,n]^m,\cL (\bbx^{(t)},\cdot,\lambda^{(t)}),\bby^{(t)},\eta_y)\]
$\qquad \qquad$ $\bbz^{(t+1)} = (\bbx^{(t+1)}, \bby^{(t+1)})$ be the new action for the primal player.\\
Set $\overline \bbx = \frac{1}{T}\sum_{t=1}^T \bbx^{t}$  and $\overline\lambda = \frac{1}{T}\sum_{t=1}^T \lambda^{(t)}$ \\
{\bf for } each player $i$ {\bf do} $\quad$ round the fractional flow: $\bbx_i\bl \gets \psrr(\overline\bbx_i)$ \\
\Return $\bbx\bl = (\bbx\bl_{i})_{i \in [n]}$
\caption{Computing Approximately Optimal Flow via JDP Gradient Descent}\label{PIGD}
\end{algorithm}

\fi

\ifnum\final=0
\begin{algorithm}[]
\caption{Computing Approximately Optimal Flow via JDP Gradient Descent}\label{PIGD}
\begin{algorithmic}[0]
\INPUT  Routing Game $\Gamma = (G,\ell,\bs)$; privacy parameters $(\eps, \delta)$; failure probability $\beta$
\OUTPUT $\bbx_i\bl$, a $s_i = (s_i^1, s^2_i)$ flow for each player $i \in [n]$
\Procedure {\pigd}{$\Gamma,\epsilon,\delta, \beta$}
\State Define the following quantities:
$$
T \gets \Theta\left( \frac{\epsilon n\sqrt{m}}{\log(mn/\beta) \sqrt{\log(1/\delta)}} \right)
\qquad \eps' \gets \eps/\sqrt{8T\ln(1/\delta)}
\qquad \eta_y \gets \frac{D_y}{G_y \sqrt{T}} \qquad \eta_x \gets \frac{D_x}{G_x \sqrt{T}}
$$
$$
\qquad G_y \gets \sqrt{ (m-1)(\gamma+1)^2 + (\gamma+1 + 2m)^2}\qquad D_y \gets n\sqrt{m} 
$$
$$
\qquad G_x \gets 2m\sqrt{n} \qquad D_x \gets\sqrt{mn}  
$$
\State Initialize: $\bby^{(1)} \in [0,n]^m$ and $\bbx^{(1)} \in
\cF^R(\bs)$. Let $\bbz^{(1)} \gets (\bbx^{(1)}, \bby^{(1)} )$ 

\State Define the quality score $q : \cG(\bs) \times  \left((+, -) \times E \right) \to
\R$:
\[
f_e(\bbz) \gets \sum_i x^i_e - y_e \qquad q(\bbz,  (+, e)) \gets + f_e(\bbz) \qquad q(\bbz,  (-, e)) \gets -f_e(\bbz).
\]
\For{$t = 1, \cdots, T$}

\State Let $(\bullet^{(t)}, e^{(t)}) \gets \cM_{E}(\bs, q, \epsilon')$ (The Exponential Mechanism)
\State Approximate best-response for the dual player $\lambda\utt$:
\If{$\bullet\utt = +\quad$} $\quad \lambda_{e^{(t)}}^{(t)} \gets -2m$
\Else $\quad \lambda_{e^{(t)}}^{(t)} \gets + 2m$
\EndIf
\For{$e' \in E\setminus\{e^{(t)}\} \quad$}
$\quad\lambda_{e'}^{(t)} = 0$
\EndFor
\State Gradient descent update on the primal:

\State Take a step to improve the individual flow variables $\bbx^{(t)}$:
\[\bbx^{(t+1)} \gets \gd(\cF(\bs),\cL (\cdot,\bby^{(t)},\lambda^{(t)}),\bbx^{(t)},\eta_x)\]
\State Take a step to improve the congestion variables $\bby^{(t)}$:
\[ \bby^{(t+1)}\gets\gd([0,n]^m,\cL (\bbx^{(t)},\cdot,\lambda^{(t)}),\bby^{(t)},\eta_y)\]
\State Let $\bbz^{(t+1)} = (\bbx^{(t+1)}, \bby^{(t+1)})$ be the new action for the primal player.
\EndFor
\State $\overline \bbx = \frac{1}{T}\sum_{t=1}^T \bbx^{t}$  and $\overline\lambda = \frac{1}{T}\sum_{t=1}^T \lambda^{(t)}$
\For{each player $i$}
 round the fractional flow: $\bbx_i\bl \gets \psrr(\overline\bbx_i)$
\EndFor
\Return $\bbx\bl = (\bbx\bl_{i})_{i \in [n]}$
\EndProcedure
\end{algorithmic}
\end{algorithm}

\fi

\subsection{Privacy of the JDP Gradient Descent Algorithm} \label{sec:pgdprivacy}

We will use Lemma \ref{lem:billboard} (the \emph{billboard lemma}) to
prove that $\pigd$ satisfies joint differential privacy.  We first
show that the sequence of plays by the dual player satisfies standard
differential privacy.  \rynote{I moved the Billboard Lemma earlier
  because we used that is an earlier privacy proof}


\begin{lemma}
  \label{lem:dualDP}
  The sequence $\{\lambda\utt\}_{t=1}^T$ in
  $\pigd(\Gamma,\eps,\delta,\beta)$ satisfies $(\eps,
  \delta)$-differential privacy in the reported types $\bs$ of the
  players.
\end{lemma}

\begin{proof}
\jnote{It's great to defer privacy preliminaries to the end.  But people should still be able to read the paper straight through.  Can we just write exactly what distribution we sample from, then state the key lemmas (dp+accuracy) and say that they follow from standard analysis of the exp mech?  Then in the appendix we can include a proof for completeness.}
At each iteration of the main for-loop, we use the exponential
mechanism with quality score $q$ to find which edge $e$ has the most
violated constraint in \eqref{eq:violation}. By Lemma
\ref{lem:exp_private}, each tuple\footnote{Recall that $\bullet\in \{+,-\}$ indicating whether $\sum_i x_{i, e} > y_e$ or $\sum_i x_{i, e} < y_e$} $(\bullet\utt, e^{(t)})$is $\eps'$-differentially private.
Note that the dual strategy $\lambda\utt$ is simply a post-processing
function of the tuple $(\bullet\utt, e\utt)$, and by Lemma~\ref{lem:post},
we know that $\lambda\utt$ is $\eps'$-differentially private. By the
composition theorem for differential privacy (Lemma~\ref{lem:advanced_comp}), we know that the
sequence of the dual plays $\lambda\utt$ satisfies $(\eps,
\delta)$-differential privacy, with the assignment of $\eps'$ in $\pigd$.

\end{proof}

We are now ready to show that our algorithm satisfies joint
differential privacy.

\begin{theorem}
$\pigd(\Gamma,\eps,\delta,\beta)$ given in Algorithm~\ref{PIGD} is $(\eps,\delta)$-jointly
  differentially private.
\label{thm:pigd_jdp}
\end{theorem}
\begin{proof}
In order to establish joint differential privacy using the Billboard
Lemma (Lemma~\ref{lem:billboard}), we just need to show that the
output solution $\{\bbx_i\}$ for each player $i$ is just a
function of the dual plays $\{\lambda\utt\}$ and $i$'s private data.

Note that initially, each player $i$ simply sets $\bbx_i^{(1)}$ to be
a feasible flow in the set $\cF^R(\bs_i)$, which only depends on $i$'s
private data.

Then at each round $t$, the algorithm updates the vector $\bbx\utt$
using the gradient:
$$\nabla_{\bbx} \cL (\bbx, \bby\utt ,\lambda^{(t)}) = \left(\left( -
\lambda_e^{(t)} \right)_{e \in E} \right)_{i\in [n]}.$$

The gradient descent update for $\bbx\utt$ is
\begin{align*}
\bbx^{(t+1)} &= \Pi_{\cF^R(\bs)} \left[\bbx\utt - \eta_x
  (\lambda\utt)_{i\in [n]}\right]\\
&= \arg\min_{\bbx\in \cF^R(\bs)} \left\|\bbx -
\left(\bbx\utt - \eta_x(\lambda\utt)_{i\in[n]} \right) \right\|_2^2\\
&= \arg\min_{\bbx\in\cF^R(\bs)} \sum_{i\in [n]}\left( \sum_{e\in E}\left\|x_{i, e} - (x\utt_{i,e} - \eta_x\lambda\utt_e)\right\|_2^2\right)
\end{align*}

Note that this update step can be decomposed into $n$ individual
updates over the players:
\[
\bbx^{(t+1)}_i = \arg\min_{\bbx\in \cF^R(\bs_i)} \sum_{e\in E} \left\|
x_{i,e} - \left(x\utt_{i,e} -\eta_x\lambda_e\utt \right) \right\|_2^2
\]

Since such an update only depends on the private data of $i$ and also
the sequence of dual plays $\{\lambda\utt\}$, we know that $\{\bbx\utt\}$
satisfies $(\eps, \delta)$-joint differential privacy by the Billboard Lemma.

Finally, the output integral solution $\bbx_i$ to each player $i$ is
simply a sample from the distribution induced by the average play of
$i$: $\overline\bbx$. Therefore, we can conclude that releasing
$\bbx_i$ to each player $i$ satisfies $(\eps, \delta)$-joint
differential privacy.
\rynote{Does the Round step in the algorithm mean sample $x$ uniformly at random from the $x^{(t)}$'s? Or are we really rounding in some other way?  Are the $\bbx^{(t)}$'s integral?}
\snote{defined now}

\end{proof}

\subsection{Utility of the JDP Gradient Descent Algorithm} \label{sec:pgdutility}
We now establish the accuracy guarantee of the integral flow $\bbx\bl$ computed by the procedure $\pigd$.  First, consider the average of the actions taken by both players over the $T$ rounds of the algorithm $\pigd$: $\overline\bbz = \frac{1}{T} \sum_{t=1}^T \bbz^{(t)}$ and $\overline \lambda = \frac{1}{T} \sum_{t=1}^T \lambda^{(t)}$.  Recall that the \emph{minimax value} of the Lagrangian game is defined as
\begin{equation}
\max_{\lambda\in \mathbb{R}^m} \min_{(\bbx, \bby) \in \cG^R(\bs)} \cL(\bbx, \bby, \lambda) =
\min_{(\bbx, \bby)\in \cG^R(\bs)} \max_{\lambda\in \mathbb{R}^m} \cL(\bbx, \bby, \lambda) = \OPT^R(\bs).
\label{eq:minimax1}
\end{equation}

Thus, in order to show that $\overline \bbz$ is a flow with nearly optimal cost (i.e.~cost not much larger than $\OPT^R(\bs)$), it suffices to show that $(\overline \bbz, \overline \lambda)$ are a pair of ``approximate minimax strategies''.  That is, each player is guaranteeing itself a payoff that is close to the value of the game.  Formally, $(\overline \bbz, \overline \lambda)$ is a pair of \emph{$\cR$-approximate minimax strategies} if
$$
\forall \bbz',\; \cL(\overline \bbz, \overline \lambda) \leq \cL(\bbz', \overline \lambda) + \cR \qquad \textrm{and} \qquad \forall \lambda',\; \cL(\overline \bbz, \overline \lambda) \geq \cL(\overline \bbz, \lambda') - \cR.
$$
\jnote{It's not immediately clear from this why approximate minimax guarantees something close to the value of the game.  We should add an explanation of that for people unfamiliar with these aguments..}

Looking ahead, using the properties of gradient descent, we can show that $(\overline \bbz, \overline \lambda)$ are a pair of $\cR$-approximate minimax strategies for a bound $\cR$ that will grow with the norm of the dual player's actions, i.e.~$\|\lambda^{(t)}\|_2$.  Thus, in $\pigd$, the dual player's action is chosen to have bounded norm (at most $2m$), in order to ensure $\cR$ is relatively small.  However, from~\eqref{eq:minimax1} it's not clear that the optimal dual strategy has small norm, so restricting the norm of the dual player's actions might change the value of the game.  However, we show that restricting the norm of the dual player's action does not change the value of the game.

\jnote{In this paragraph we take $\lambda^*$ to be an unrestricted maxmin strategy.  So the maxmin strategy promised in Lem 4.5 might not be equal to $\lambda^*$.  So I called it $\lambda^*_{\cB}$.}
Let $(\bbz^*, \lambda^*)$ be a pair of (exact) minimax strategies in the Lagrangian game.  By strong duality, we know that $$\cL(\bbz^*, \lambda^*) = \OPT^R(\bs)$$ and $\bbz^*$ is an optimal and feasible solution. We now reason about the \emph{restricted Lagrangian game}, in which the dual player's action is restricted to the space $\cB = \{\lambda\in \mathbb{R}^{m} \mid \|\lambda\|_1 \leq 2m\} \subseteq \R^m$.  The next lemma states that even when the dual player's actions are restricted, then $\bbz^*$ is still a minimax strategy for the primal player.  That is, the primal player cannot take advantage of the restriction on the dual player to obtain a higher payoff. 
\jnote{Do we want an $L1$ bound here?  The regret bound for \gd~is naturally stated in terms of $L2$.  Maybe we should just use $L2$ here to match?}
\begin{lemma}
There exists a dual strategy $\lambda^*_{\cB} \in \cB$ such that $(\bbz^*, \lambda^*_{\cB})$ is a pair of (exact) minimax strategies for the restricted Lagrangian game.
\end{lemma}


\begin{proof}
Since $\bbz^*$ is an (exact) minimax strategy for the (unrestricted) Lagrangian game, we know that for any $\lambda\in \cB$
  \[
  \cL(\bbz^*, \lambda) = c(\bby^*) = \OPT^R(\bs).
  \]

  Let $\bbx' \in \cF^R(\bs)$ and $\bby\in [0, n]^m$ be different flows
  such that $\bbx' \neq \bbx^*$ and $\bby' \neq \bby^*$. We want to show that
\[\max_{\lambda\in \cB}
\cL(\bbx, \bby, \lambda) \geq \max_{\lambda\in \cB}\cL(\bbz^*,\lambda):= \cL(\bbz^*,\lambda^*_{\cB}).\]

If we have $y'_e = \sum_{i = 1}^n x'_{i,e}$ for all $e \in E$, then
$$ \max_{\lambda\in \cB} \cL(\bbx, \bby, \lambda) = c(\bby) \geq c(\bby^*).$$

Suppose there is some edge such that $y'_e \neq \sum_{i = 1}^n
x'_{i,e}$, then we define $\Delta := \|\bby' -
\sum_{i=1}^n\bbx'_i\|_\infty$.  
With the cost function in terms of the individual flow variables in
\eqref{eq:ind_cost} we know that
\[c(\bby) \geq  \phi(\bbx) - \frac{1}{n}\sum_{e\in E}\Delta\cdot \ell_e(n)\geq \phi(\bbx) - m\cdot\Delta \geq c(\bby^*) - m\cdot\Delta.\]

Note that the dual player can set $\lambda_e = 2m$ for $\sum_{i=1}^n
x'_{ie} - y'_e>0 $ or $\lambda_e = -2m$ for $\sum_{i=1}^n x'_{ie} - y'_e <
0$ for the maximally violated edge $e$:

\[
\max_{\lambda\in \cB} \cL(\bbx, \bby, \lambda) = c(\bby) + 2m\cdot \Delta \geq  c(\bby^*) + \cdot \Delta(2m - m) > \OPT^R(\bs).
\]
Therefore, any infeasible $(\bbx, \bby)$ would suffer loss at least
$\OPT^R(\bs)$ in the worst case over the dual strategy space. It follows that
$(\bbz^*,\lambda^*)$ is a minimax strategy.

Since both players' action spaces $\cG^R(\bs)$ and $\cB$ are compact, then
there exists a minimax strategy $(\bbz^*,\lambda^*_{\cB})$ of the
restricted Lagrangian game.
\end{proof}

Using the previous lemma, we know that the value of the restricted game is the same 
\begin{lemma}
Let $(\bbz, \lambda)$ be a pair of $\cR$-approximate
minimax strategy of the restricted Lagrangian game, and $\bbz = (\bbx, \bby)$.  Then the fractional solution $\bbx$
satisfies
\[
\phi(\bbx) \leq \OPT^R(\bs) + 4\cR.
\]
\end{lemma}
\jnote{I feel like Lems 4.3/4.4 just say that restricting the game doesn't change the value.  Can we just replace them with one lemma that says that?}

\begin{proof}
We will first bound the constraint violation in $(\bbx, \bby)$.
Let $e'\in \arg\max_{e\in E} |\sum_i x_{i, e} - y_e|$ be an edge where
the constraint is violated the most, and let $\Delta = |\sum_i x_{i,
  e'} - y_{e'}|$. Consider the dual strategy $\lambda' \in \cB$ such that
\[
\lambda'_{e'} = \begin{cases}
  -2m &\mbox{ if } \sum_i x_{i, e'} - y_{e'} \geq 0\\
  2m &\mbox{ otherwise}
\end{cases}
\]
and $\lambda'_e = 0$ for all $e\neq e'$. Now compare the payoff values
$\cL(\bbx, \bby, \lambda)$ and $\cL(\bbx, \bby, \lambda')$. By the property
of $\cR$-approximate equilibrium and letting $((\bbx^*,\bby^*),\lambda^*)$ be the exact equilibrium, we have
\begin{align*}
\OPT^R(\bs) & - \cR  = \cL(\bbx^*,\bby^*,\lambda^*) - \cR \leq \cL(\bbx,\bby,\lambda^*) - \cR \\
& \leq \cL(\bbx, \bby, \lambda) \leq \cL(\bbx^*,\bby^*,\lambda) + \cR \leq \OPT^R(\bs) + \cR \\
\implies \qquad & \OPT^R(\bs)  - \cR \leq \cL(\bbx, \bby, \lambda) \leq \OPT^R(\bs) + \cR
\end{align*}
and
\[
\cL(\bbx, \bby, \lambda') \leq \OPT^R(\bs) + 2\cR.
\]
Since $(\bbx, \bby)$ violates equality constraint on each edge by at
most $\Delta$, we know that
\[
c(\bby) \geq \phi(\bbx) - \frac{1}{n}\cdot \sum_{e\in E} \left|\sum_i x_{i, e} - y_e\right|\cdot \ell_e(n) \geq \OPT^R(\bs) - m\Delta.
\]
Also, the penalty incurred by $\lambda'$ is at least
\[
\sum_{e\in E} \lambda'_e \left(y_e - \sum_i x_{i, e} \right) = 2m\cdot \Delta.
\]
Therefore, we could bound
\[
\cL(\bbx, \bby, \lambda') \geq \OPT^R(\bs) + m\cdot \Delta.
\]
It follows that $\Delta \leq 2\cR/m$.

Next we will show the accuracy guarantee of $\bbx$. Consider an
all-zero strategy for the dual player $\lambda''$, that is $\lambda_e'' = 0$
for each $e\in E$. We know such a deviation will not increase the
payoff by more than $\cR$:
\[
\cL(\bbx, \bby, \lambda'') \leq \cL(\bbx, \bby, \lambda) + \cR \leq \OPT^R(\bs) +2\cR,
\]
and also $\cL(\bbx, \bby, \lambda'') = c(\bby)$, so we must have
\[
c(\bby) \leq \OPT^R(\bs) + 2\cR.
\]
Now we could give the accuracy guarantee for the cost of the individual flows $\bbx$:
\[
\phi(\bbx) \leq c(\bby) + \frac{1}{n}\sum_{e\in E} \Delta \cdot \ell_e(n)\leq
\OPT^R(\bs) + 2\cR + 2(\cR/m)\cdot m = \OPT^R(\bs) + 4\cR.
\]
This completes the proof of the lemma.
\end{proof}

The previous discussion shows that if $(\overline \bbz, \overline \lambda)$ is a pair of approximate minimax strategies for the restricted Lagrangian game, then $\overline \bbx$ represents an approximately optimal flow.  In the remainder of this section, we show that $(\overline \bbz, \overline \lambda)$ will be such a pair of strategies.    To do so, we use a well known result of Freund and Schapire~\cite{FS96}, which states that if $\overline \bbz$ and $\overline \lambda$ have ``low regret,'' then they are a pair of approximate minimax strategies.  ``Regret'' is defined as follows.
\begin{definition} Given a sequence of of actions $\{\bbz\utt\}$ and $\{\lambda\utt\}$ in the Lagrangian game, we define \emph{regret} for each player as:
\begin{align*}
  \cR_\bbz &\equiv \frac{1}{T}\sum_{t=1}^T \cL(\bbz^{(t)}, \lambda^{(t)}) - \min_{\bbz\in \cG(\bs)}\frac{1}{T}\sum_{t=1}^T \cL(\bbz, \lambda^{(t)})\\
  \cR_\lambda &\equiv  \max_{\lambda\in \cB} \frac{1}{T} \sum_{t=1}^T \cL(\bbz^{(t)}, \lambda) - \frac{1}{T}\sum_{t=1}^T \cL(\bbz^{(t)}, \lambda^{(t)})
\end{align*}
\end{definition}

Given this definition, the result of~\cite{FS96} can be stated as follows.
\begin{theorem}[\cite{FS96}]\label{thm:avgeq}
If $(\overline \bbz, \overline \lambda)$ is the average of the primal and dual players' actions in $\pigd$, then $(\overline \bbz, \overline \lambda)$ is a pair of $(\cR_\bbz + \cR_\lambda)$-approximate minimax strategies of the restricted Lagrangian game.
\end{theorem}

Given the previous theorem, our goal is now roughly to show that $\overline \bbz$ and $\overline \lambda$ have low regret.  To do so, we need to analyze the regret properties of the gradient descent procedure, as well as the additional regret incurred by the noise added to ensure joint differential privacy.  

Specifically, the gradient descent procedure $\gd$ satisfies the following regret bound.
\jnote{Is this really from Zinkovich?}
\jnote{I rewrote lem 4.6 so assumptions come first and conclusions come second.  Yoda may be wise and powerful, but his papers don't get into EC.}
 \begin{lemma}[\cite{Z03}]\label{lem:gd_regret}
  Fix the number of steps $T \in \N$.  Let $\hat\cD$ be a convex and closed set with bounded diameter, i.e.~for every $\omega, \omega' \in \cD,$
$$
|| \omega - \omega ' ||_2 \leq D.
$$
Let $r^1,\dots,r^T$ be a sequence of differentiable, convex functions with bounded gradients, i.e.~for every step $t \in [T]$,$$
||\nabla r^t ||_2 \leq G.
$$
Let $\eta = \frac{D}{G\sqrt{T}}$ and $\omega^{0} \in \cD$ be
arbitrary.  Then if we compute $\omega^{1},\dots,\omega^{T} \in \cD$
according to the rule $\omega^{t+1} \gets \gd(\cD,
r^t,\omega^{t}, \eta^t)$, the sequence $\{\omega^1, \ldots
,\omega^{T}\}$  satisfies
\begin{equation}
R^T(\gd) : = \sum_{t=1}^T r^t(\omega^{t}) - \min_{\omega \in \cD} \left\{ \sum_{t=1}^T r^t(\omega) \right\} \leq G D \sqrt{T}
\label{eq:gd_regret}
\end{equation}
\end{lemma}

We can now use this regret bound for $\gd$ to give a regret bound for the private gradient descent procedure $\pigd$.
\begin{lemma} \label{lem:fractionalquality}
Fix $\eps,\delta,\beta>0$.  If $(\overline \bbz, \overline\lambda)$ is the average of the primal and dual players' actions in $\pigd(\Gamma,\eps,\delta,\beta)$, then with probability at least $1-\beta$, $(\overline \bbz, \overline \lambda)$ are a pair of $\cR$-approximate minimax strategies in the restricted Lagrangian game, for
\[
\cR = \tilde O\left( \frac{\sqrt{n}m^{5/4}}{\sqrt{\eps}} \right)
\]
\end{lemma}

\begin{proof}
In light of Theorem~\ref{thm:avgeq}, we know $\cR = \cR_z +
\cR_\lambda$, so we just need to bound the regrets for both players.  For the flow
player $\bbz$, we will bound the regrets of $\bbx$ and $\bby$ separately by
invoking the regret bound of~\cite{Z03} given in Lemma \ref{lem:gd_regret}.

We define $G_y$ such that $\forall t \in [T]$ we have $|| \nabla_\bby
\cL (\bbz,\lambda^{(t)}) ||_2 \leq G_y$ and $D_y$ such that $\forall
\bby,\bby' \in [0,n]^m$ we have $||\bby-\bby'||_2\leq D_y$.  We define
corresponding quantities for $G_x$ and $D_x$.  It suffices to set
these values in the following way:
$$
G_y : = \sqrt{ (m-1)(\gamma+1)^2 + (\gamma+1 + 2m)^2}\qquad D_y := n\sqrt{m}
$$
$$
G_x : = 2m\sqrt{n} \qquad D_x : =\sqrt{mn}
$$
Using \eqref{eq:gd_regret} we have the following bound on the regret \ifnum\lip=0 with $\gamma = o(n)$\fi.
\begin{align}
  \cR_z&\leq 1/\sqrt{T} \left(G_y\cdot D_y + G_x\cdot D_x
  \right) \nonumber\\ 
  &\leq \frac{n\sqrt{m}}{\sqrt{T}}\cdot \left(\sqrt{m(\gamma+1)^2+(\gamma+1+2m)^2} +2m \right) \nonumber \\
  & = O\left(\frac{nm^{3/2}}{\sqrt{T}} \right)
  \label{eq:pigd_regret}
\end{align}
with the following step sizes:
\begin{equation}
 \eta_x := \frac{D_x}{G_x \sqrt{T}} \qquad \eta_y := \frac{D_y}{G_y \sqrt{T}}
 \label{eq:eta_y}
\end{equation}

Now we bound the regret for the dual player. Note that each agent
could only affect the quality score $q$ of each edge by 1. By the utility
guarantee of the exponential mechanism stated in Lemma
\ref{lem:exp_utility} we know that with probability at least $1 -
\beta/T$, at round $t$
\begin{equation}
\max_{(\bullet, e)\in \{\pm\}\times E } \left|q(\bbz\utt, (\bullet, e)) - q(\bbz\utt,
(\bullet\utt, e\utt)\right| \leq \frac{2(\log(2mT/\beta))}{\eps'}
\label{eq:error_exp}
\end{equation}
We condition on this level of accuracy for each round $t$, which holds except with probability $\beta$.

Also, at each round $t$, a best response for the dual player is to put
weight $\pm 2m$ on the edge with the most violation, so we can bound
the regret:
\begin{align*}
\cR_\lambda &= \frac{1}{T} \sum_{t=1}^T \left[\max_{\lambda\in \cB}\cL(\bbz\utt, \lambda) - \cL(\bbz\utt, \lambda\utt)  \right]\\
&\leq \frac{1}{T} \sum_{t=1}^T 2m \cdot \frac{2(\log(2mT/\beta))}{\eps'}\\
& = 2m \cdot \frac{2(\log(2mT/\beta))}{\eps'}
\end{align*}
For $T = \Theta\left(\frac{\eps n\sqrt{m}}{\log(mn/\beta) \sqrt{\log(1/\delta)}}\right)$, we know that
\begin{align*}
\cR &= \cR_z + \cR_\lambda =O\left(\frac{ n m^{3/2}}{\sqrt{T}} + \frac{m\log(mT/\beta) \sqrt{T \log(1/\delta)}}{\epsilon} \right) \\
&= O\left( \frac{\sqrt{n} m^{5/4}}{\sqrt{\eps}} \cdot \polylog(1/\delta, 1/\beta, n, m)\right)
\end{align*}
\end{proof}


The previous lemma shows that the fractional solution has nearly optimal cost.  The last thing we need to do is derive a bound on how much the rounding procedure $\psrr$ increases the cost of the final integral solution.

\begin{lemma} \label{lem:roundingquality}
 Let $\overline\bbx$ be any feasible fractional solution to the convex
 program~\eqref{eq:convex_program}, and let $\bbx\bl$ be an integral solution obtained by the rounding procedure
 $\psrr(\overline\bbx)$. Then, with probability at
 least $1 - \beta$,
 \[
 \phi(\bbx\bl) \leq \phi(\overline \bbx) + m(\gamma+1)\sqrt{2{n\ln(m/\beta)}}.
 \]
\end{lemma}

\begin{proof}
From the analysis of~\cite{RT87} (in Theorem 3.1 of the source), we know that with probability at least $1- \beta$,
 \[
 \max_{e\in E} \left[\sum_i x_{i,e}\bl - \sum_i \overline{x}_{i, e}\right] < \sqrt{{2n\ln(m/\beta)}} \equiv W.
 \]
 Finally, we could bound the difference between the costs $\phi(\bbx\bl)$ and $\phi(\overline\bbx)$
 \begin{align*}
\phi(\bbx\bl) - \phi(\overline\bbx) &\leq \frac{1}{n} \cdot \left[\sum_e x_{i,e}\bl\cdot \left( \ell_e\left(\sum_j x_{j,e}\bl \right) - \ell_e \left(\sum_j \overline x_{j,e} \right) \right) + \sum_{e} W \cdot \ell_e \left(\sum_j  x\bl_{j,e} \right) \right] \\
& \leq \frac{1}{n} \cdot \left(mn\gamma W + m n W \right) =
 mW(\gamma+1).  \end{align*} This completes the proof.
\end{proof}

Combining Lemma~\ref{lem:fractionalquality} and Lemma~\ref{lem:roundingquality} we obtain our desired bound on the quality of the joint differentially private integral solution.
\begin{theorem}
Let $\Gamma = (G,\ell,\bs)$ be a routing game and $\eps,\delta,\beta>0$ be parameters. If $\bbx\bl$ is the final integral solution given by \pigd$(\Gamma,\eps,\delta,\beta)$, then with probability at least $1 -\beta$, the cost of $\bbx\bl$ satisfies
\[
\phi(\bbx\bl) \leq \OPT(\bs) + \tilde O \left( \frac{\sqrt{n} m^{5/4}}{\sqrt{\eps}} + m\sqrt{n} \right)
\]
i.e. $\bbx\bl$ is an $\alpha$-approximate optimal flow for $\alpha= \tilde O \left( \frac{ \sqrt{n}m^{5/4}}{\sqrt{\eps}} \right)$.
\label{thm:approx_opt_flow}
\end{theorem}

\fi

\ifnum\final=0
\addcontentsline{toc}{section}{References}
\bibliographystyle{acmsmall} 
\bibliography{references}
\else
\bibliographystyle{acmsmall}
\bibliography{references}
\fi

\ifnum\appen=1
\newpage
\appendix
\section{Tools for Differential Privacy}

In this section we review the necessary privacy definitions and tools
needed for our results.  Throughout, let $\bs = (s_1,\dots,s_n) \in
\cS^n$ be a \emph{database} consisting of $n$ elements from a domain
$\cS$.  In keeping with our game theoretic applications, we refer to
the elements $s_1,\dots,s_n$ as \emph{types} and each type belongs to
a \emph{player} $i \in [n]$.

%

We first state a general lemmas about differential privacy.
\begin{lemma}[Post-Processing \cite{DMNS06}]
  Given a mechanism $\cM: \cS^n \to \mathcal{O}$ and some
  (possibly randomized) function $p: \mathcal{O} \to \mathcal{O}'$
  that is independent of the players' types $\mathbf{s} \in
  \cS^n$, if $\cM(\mathbf{s})$ is $(\eps,
  \delta)$-differentially private then $p(\cM(\mathbf{s}))$ is
  $(\eps, \delta)$-differentially private.
\label{lem:post}
\end{lemma}

\subsection{The Laplace Mechanism}
\label{laplace}
We will use the Laplace Mechanism, which was introduced by Dwork et al \cite{DMNS06} to answers a vector-valued query $f : \cS^n \to \R^k$.

The Laplace Mechanism depends on the notation of \emph{sensitivity}---how much a function can change when a single entry in its input is altered.
\begin{definition}[Sensitivity]
 The sensitivity $\Delta f$ of a function $f : \cS^n \to \R^k$ is defined as
 $$
 \Delta f = \max_{ i \in [n], (s_i,\bs_{-i}) \in \cS^n, s_i' \in \cS } \left\{ ||f(s_i,\bs_{-i}) - f(s_i',\bs_{-i})||_1 \right\}.
 $$
\end{definition}

\ifnum\final=1
\begin{algorithm}
\SetAlgoNoLine
$\ML(\bs,f, \eps)$\;
\KwIn{Database $\bs \in \cS^n$, query $f: \cS^n \to \R^k$, and privacy parameter $\eps$.}
\KwOut{Approximate to $f(\bs)$}
 Set $\hat a = f(\bs) + Z \qquad Z= (Z_1, \cdots, Z_k) $ and $Z_i \sim \Lap(\Delta f/\eps)$\\
\Return $\hat a$.
 \caption{Laplace Mechanism \cite{DMNS06}}\label{lap_mech}
\end{algorithm}
\fi

\ifnum\final=0
\begin{algorithm}
 \caption{Laplace Mechanism \cite{DMNS06}}\label{lap_mech}
\begin{algorithmic}[0]
\INPUT : Database $\bs \in \cS^n$, query $f: \cS^n \to \R^k$, and privacy parameter $\eps$.
\Procedure {$\ML$} {$\bs,f, \eps$}
\State Set $\hat a = f(\bs) + Z \qquad Z= (Z_1, \cdots, Z_k) $ and $Z_i \sim \Lap(\Delta f/\eps)$\\
\Return $\hat a$.
\EndProcedure
\end{algorithmic}
\end{algorithm}
\fi

\begin{lemma}[\cite{DMNS06}]
 \label{lem:lap_privacy}
 The Laplace Mechanism $\ML$ is $\eps$-differentially private.
\end{lemma}
\begin{lemma}[\cite{DMNS06}]
 \label{lem:lap_utility}
 The Laplace Mechanism $\ML(\bs,f,\eps)$ produces output $\hat a$ such that with probability at least $1-\beta$ we have
 $$
 ||f(\bs)  - \hat a||_\infty \leq \log\left(\frac{k}{\beta}\right)\left(\frac{\Delta f}{\eps} \right) 
 $$
\end{lemma}

\subsection{The Exponential Mechanism} \label{sec:exp_mech}
We now present an algorithm introduced by \cite{MT07} that is differentially private called the exponential mechanism.  Let us assume that we have some finite outcome space $\mathcal{O}$ and a quality score $q: \cS^n \times \mathcal{O} \to \R$ that tells us how good the outcome $o \in \mathcal{O}$ is with the given database $\bs \in \cS^n$.  We define the sensitivity of $q$ as the maximum over $o \in \mathcal{O}$ of the sensitivity of $q(\cdot; o)$.  Specifically,
$$
\Delta q  = \max_{\begin{array}{lr} o \in \mathcal{O},   \bs, \bs' \in \cS^n\end{array}} \left\{|q(\bs,o) - q(\bs', o)| \right\}  \qquad \text{ for neighboring } \bs, \bs'  
$$
\ifnum\final=1
\begin{algorithm}[h]

$\cM_E(\bs,q, \eps)$\;
\KwIn{Database $\bs \in \cS^n$, quality function $q: \cS^n \times \mathcal{O} \to \R$, and privacy parameter $\eps$.}
\KwOut{Outcome $o \in \mathcal{O}$ with probability proportional to} 
$$
\exp\left( \frac{\eps q(\bs,o)}{2 \Delta q} \right)
$$ 
\caption{Exponential Mechanism \cite{MT07}}\label{exp_mech}
\end{algorithm}
\fi

\ifnum\final=0
\begin{algorithm}[h]
\caption{Exponential Mechanism \cite{MT07}}\label{exp_mech}
\begin{algorithmic}[0]
\INPUT : Database $\bs \in \cS^n$, quality function $q: \cS^n \times \mathcal{O} \to \R$, and privacy parameter $\eps$.
\Procedure {$\cM_E$} {$\bs,q, \eps$}
\State Output $o \in \mathcal{O}$ with probability proportional to 
$$
\exp\left( \frac{\eps q(\bs,o)}{2 \Delta q} \right)
$$ 
\EndProcedure
\end{algorithmic}
\end{algorithm}
\fi
\begin{lemma}[\cite{MT07}]
The Exponential Mechanism $\cM_E$ is $\eps$-differentially private.  
\label{lem:exp_private}
\end{lemma}
We then define the highest possible quality score with database $d$ to be $OPT_q(\bs) = \max_{o \in O}\{q(\bs,o) \}$.  We then obtain the following proposition that tells us how close we are to the optimal quality score.  
\begin{lemma}[\cite{MT07}]
We have the following utility guarantee from the Exponential Mechanism $\cM_E$: with probability at least $1-\beta$ and every $t > 0$,
$$
q(\bs, \cM_E(\bs,q, \eps) ) \geq OPT_q(\bs) - \frac{2\Delta q}{\eps} \left(\log |O| + t \right)
$$
\label{lem:exp_utility} 
\end{lemma}

\subsection{Composition Theorems}
\label{sec:comp}
Now that we have given a few differentially private algorithms, we
want to show that differentially private algorithms can compose
``nicely'' to get other differentially private algorithms. We will
need to use two composition theorems later in this paper. The first
shows that the privacy parameters add when we compose two
differentially private mechanisms, and the second from \cite{DRV10}
gives a better composition guarantee when using many adaptively chosen
mechanisms.
\begin{lemma}
  If we have one mechanism $M_1: \cS^n \to O$ that is $(\eps_1,
  \delta_1)$-differentially private, and another mechanism $M_2:\cS^n
  \times O \to R$ is $(\eps_2,\delta_2)$-differentially private in
  its first component, then $M: \cS^n \to R$ is
  $(\eps_1+\eps_2,\delta_1+\delta_2)$ differentially private
  where
$$
M(\bs) = M_2(\bs,M_1(\bs)).
$$
\label{lem:comp}
\end{lemma}

If we were to only consider the previous composition theorem, then the
composition of $m$ mechanisms that are $\eps$-differentially
private mechanisms would lead to a $m\eps$-differentially private
mechanism. However, the next theorem says that we can actually get
$(\eps',\delta)$-differential privacy where $\eps' =
O(\sqrt{m}\eps)$ if we allow for a small $\delta>0$. This theorem
also holds under the threat of an adversary that uses an adaptively
chosen sequence of differentially private mechanisms so that each can
use the outputs of the past mechanisms and different datasets that may
or may not include an individual's data. See \cite{DRV10} for further
details.

 
\begin{lemma}[$m$-Fold Adaptive Composition \cite{DRV10}]
  Fix $\delta>0$. The class of $(\eps',\delta')$ differentially
  private mechanisms satisfies $(\eps,m\delta' + \delta)$
  differential privacy under $m$-fold adaptive composition for
$$
\eps ' = \frac{\eps}{\sqrt{8 m \log(1/\delta)}}.
$$
\label{lem:advanced_comp}
\end{lemma}

We also include a proof of Lemma~\ref{lem:djcomp}, the composition of
a differentially private algorithm with another joint differentially
private algorithm is differentially private. 

\begin{proof}[\ifnum\final=0 Proof \fi of Lemma~\ref{lem:djcomp}]
Let $S \subseteq O$, $i \in [n]$, and consider data $\bs \in \cS^n$
and $\bs' = (s_i',\bs_{-i})$ for $s_i' \in \cS$.\jnote{I'm allergic to
  measure theory, but technically $S$ needs to be measurable
  right?}\snote{maybe defer this to appendix} We have
\begin{align*}
\Prob{}{M(\bs) \in S}=&  \int_{\cX^n} \Prob{}{M_D(\bbx) \in S} \cdot \Prob{}{M_J(\bs) = \bbx }d\bbx \\
& = \int_{\cX^{n-1}}\left[\int_{\cX}\Prob{}{M_D(x_i,\bbx_{-i}) \in S} \cdot \Prob{}{M_J(\bs) = \bbx}dx_i \right] d\bbx_{-i}
\end{align*}
We then use the fact that, since $M_D$ satisfies $\eps_D$-differential privacy, we know that for any fixed $x_i' \in \cX$, it holds that $\Prob{}{M_D(x_i,\bbx_{-i}) \in S} \leq \min\{e^{\eps_{D}} \cdot \Prob{}{M_D(x_i',\bbx_{-i}) \in S}, 1\}$.  We let $P_{x_i', \bbx_{-i}}$ denote the RHS of this inequality.
\begin{align*}
 \Prob{}{M(\bs) \in S}\leq{} & \int_{\cX^{n-1}}\left[\int_{\cX} P_{x_i', \bbx_{-i}} \cdot \Prob{}{M_J(\bs) = \bbx}dx_i\right] d\bbx_{-i} \\
={}  &  \int_{\cX^{n-1}}P_{x_i', \bbx_{-i}} \cdot \Prob{}{M_J(\bs)_{-i} = \bbx_{-i}}  d\bbx_{-i}
\end{align*}
Now we use the fact that, since $M_J$ satisfies $(\eps_J, \delta)$-joint differential privacy, we have the inequality $\Prob{}{M_J(\bs)_{-i} = \bbx_{-i}} \leq e^{\eps_J} \cdot \Prob{}{M_J(\bs')_{-i} = \bbx_{-i}} + \delta = ( \int_{\cX} e^{\eps_J} \cdot\Prob{}{M_J(\bs') = \bbx} dx_i) + \delta$.
\begin{align*}
  \Prob{}{M(\bs) \in S}\leq{} & \int_{\cX^{n-1}}P_{x_i', \bbx_{-i}} \cdot  \left[\int_{\cX} e^{\eps_J}\cdot \Prob{}{M_J(\bs') = \bbx}dx_i +\delta\right]d\bbx_{-i} \\
 \leq{} & e^{ \eps_J}\cdot\int_{\cX^{n-1}}P_{x_i', \bbx_{-i}} \cdot\left[\int_{\cX} \Prob{}{M_J(\bs') = \bbx}dx_i \right]d\bbx_{-i} + \delta \cdot \int_{\cX^{n-1}} P_{x_i', \bbx_{-i}} d\bbx_{-i}
 \end{align*}
 Using our definition of $P_{x_i', \bbx_{-i}}$ (using the first term in the min for the first term above and the second term of the min for the second term above) we can simplify as follows.
 \begin{align*}
 \Prob{}{M(\bs) \in S} \leq{} & e^{\eps_D + \eps_J}\cdot\int_{\cX^{n}}\Prob{}{M_D(x_i',\bbx_{-i}) \in S}\Prob{}{M_J(\bs') = \bbx}d\bbx + \delta
 \end{align*}
 Again we can apply the fact that, since $M_D$ is $\eps_D$-differentially private, for every $x_i \in \cX$, we have that $\Prob{}{M_D(x'_i,\bbx_{-i}) \in S} \leq e^{\eps_{D}} \cdot \Prob{}{M_D(\bbx) \in S}$.
 \begin{align*}
  \Prob{}{M(\bs) \in S} \leq{} & e^{2\eps_D + \eps_J} \cdot \int_{\cX^{n}}\Prob{}{M_D(\bbx) \in S}\Prob{}{M_J(\bs') = \bbx}d\bbx + \delta \\
 ={} & e^{2\eps_D + \eps_J} \cdot \Prob{}{M(\bs') \in S}+  \delta
\end{align*}
Since this bound holds for every neighboring pair $\bs, \bs'$, we have proven the lemma.
\end{proof}

\section{Bounding the Number of Unsatisfied Players}\label{append_sbr}

We seek to bound the number of players that are approximately
unsatisfied w.r.t. congestion $\hat\bby$ in the approximately optimal
flow $\bbx\bl$ under the routing game $\Gamma^{\hat\tau}$, where
$\hat\bby$ is the perturbed version of congestion $\bby\bl = \sum_i
\bbx_i\bl$.  First, we give a way to bound the number of unsatisfied
players for any approximately optimal flow in the routing game
$\Gamma^{\tau^*} = (G,\ell+\tau^*,\bs)$ that uses the functional
marginal-cost tolls $\tau^*(\cdot)$ given in \eqref{eq:marginal-cost}.

\begin{lemma}
  Let $\rho>0$ and $\bbx\bl$ be an $\alpha$-approximately optimal flow
  in the routing game $\Gamma$. Then the number of
  $\zeta_1(\rho)$-unsatisfied players in $\Gamma^{\tau^*}$ with respect
  to congestion $\bby\bl = \sum_{i = 1}^n\bbx_i\bl$ is bounded by
  $n\alpha/\rho$ where
  $$
  \zeta_1(\rho) = \rho +4mn\gamma  \alpha/\rho
  $$
  \label{lem:satisfied_c_prime}
\end{lemma}

\begin{proof}
Let $\bbx$ be any flow in $\cF(\bs)$. Consider the following
$\rho$-best response dynamics: while there exists some
$\rho$-unsatisfied agent $i$ (w.r.t. the true congestion $\sum_i
\bbx_i$), let $i$ make a deviation that decreases her cost the most.
Recall that we write $\OPT(\bs)$ as the optimal value for the routing
game $\Gamma$.  Note that in the tolled routing game
$\Gamma^{\tau^*}$, the potential function $\Psi$ given in
\eqref{eq:potential} satisfies $\Psi(\bbx) = n\cdot \phi(\bbx)$.

Note that $\bbx\bl$ is an $\alpha$-approximately optimal flow, so
$$
\OPT(\bs) \leq \frac{1}{n} \cdot \Psi(\bbx\bl) \leq   \OPT(\bs) + \alpha.
$$ 

Since each deviation a player made in the dynamics decreases the potential function
$\Psi(\bbx)$ by at least $\rho$, $\rho$-best response dynamics
in game $\Gamma^{\tau^*}$ starting with flow $\bbx\bl$ will terminate after
at most $n \alpha/\rho$ iterations. The resulting flow $\hat\bbx$ has all
agents $\rho$-satisfied. In the process, the congestion of each edge might
have increased or decreased by at most $n\alpha/\rho$.  For each edge
$e\in E$, the change in latency is bounded using our
$\gamma$-Lipschitz condition
\[
|\ell_e(y_e) - \ell_e(y_e')| \leq n \gamma \alpha/\rho.
\]
Furthermore, the edge toll is also $\gamma$-Lipschitz
\begin{align*}
|\tau_e^*(y_e) - \tau_e^*(y_e')| &= |( y_e - 1)(\ell_e( y_e) -
\ell( y_e - 1)) - ( y'_e - 1)(\ell_e( y'_e) -
\ell( y'_e - 1))| \\
& \leq \gamma |( y_e - 1) - ( y'_e - 1)|
\leq n \gamma \alpha/\rho.
\end{align*}
For the agents that did not deviate in the dynamics, their cost is
changed by at most $2m n\gamma \alpha/\rho$.  Since they are
$\rho$-satisfied at the end of the dynamics, this means they were
$(\rho + 4m\gamma n \alpha/\rho)$-satisfied in the beginning of the
process.\footnote{While the same path agent $i$ is taking might have
  cost lowered by $2m\gamma n \alpha/\rho$ in the dynamics, any alternate
  $(s_i, t_i)$-path might have increased its cost by
  $2m\gamma n\alpha/\rho$.}  Since the $\rho$-best response dynamics
lasts for $n \alpha/\rho$ rounds, there are at most $n \alpha/\rho$ number
of agents that deviate in the dynamics.
\end{proof}

Based on Lemma~\ref{lem:satisfied_c_prime}, we can now bound the
number of approximately unsatisfied players when we impose constant
tolls $\tau' = \tau^*(\bby\bl)$ instead of functional tolls on the
edges.

\begin{lemma}\label{lem:fun2con}
  Let $\rho > 0$, $\bbx\bl$ be an $\alpha$-approximately optimal flow
  in the routing game $\Gamma$, and $\tau' = \tau^*(\sum_i \bbx_i\bl)$
  be the vector of constant tolls. Then, the number of
  $\zeta_2(\rho)$-unsatisfied players with respect to $\bby\bl =\sum_i
  \bbx_i\bl$ in the routing game $\Gamma^{\tau'}$ is bounded by
  $n \alpha/\rho$, where
  \[
  \zeta_2(\rho) = \rho + 4m\gamma n \alpha/\rho + 2m\gamma.
  \]
\end{lemma}

\begin{proof}
  Let player $i$ be a $\zeta_1(\rho)$-satisfied player in flow
  $\bbx\bl$ under the routing game $\Gamma^{\tau^*}$. Now we argue
  that he should also be $\zeta_2(\rho)$-satisfied under the game
  $\Gamma^{\tau'}$. Suppose not. Then there exists a route $\bbx_i'$
  for player $i$ that can decrease the cost by more than
  $\zeta_2(\rho)$ under $\Gamma^{\tau'}$. Now consider the same
  deviation in game $\Gamma^{\tau^*}$. Since the functional toll on
  each edge can change by at most $\gamma$, we know that player $i$'s
  costs in $\Gamma^{\tau^*}$ and $\Gamma^{\tau'}$ differ by at most
  $2m\gamma$. This implies that the deviation $\bbx_i'$ in game
  $\Gamma^{\tau^*}$ could gain him more than $\zeta_1(\rho)$ since
  $\zeta_2(\rho) - \zeta_1(\rho) = 2m\gamma$.

  From Lemma~\ref{lem:satisfied_c_prime}, we know that the number of
  $\zeta_1(\rho)$-unsatisfied players under the routing game
  $\Gamma^{\tau^*}$ is bounded by $n\alpha/\rho$. Therefore, we know
  that the number of $\zeta_2(\rho)$-unsatisfied players under
  $\Gamma^{\tau'}$ is also bounded by $n\alpha/\rho$.
\end{proof}

Combining the previous two lemmas, we could now bound the number of
unsatisfied players in $\Gamma^{\hat \tau}$ with respect to $\bby =
\sum_{i = 1}^n \bbx_i$ with the differentially private constant tolls
$\hat\tau$.

\begin{lemma}
  \label{lem:unsatisfied}
  Let $\rho, \eps > 0$ and $\bbx\bl$ be an $\alpha$-approximately
  optimal flow in the routing game $\Gamma$. Let $\hat \tau = \tau^*(\hat\bby)$ where $\hat\bby = \pcon(\bbx\bl,\eps)$. Then with
  probability at least $1-\beta$, the number of $\zeta_\eps(\rho)$-unsatisfied
  players in $\Gamma^{\hat\tau}$ with respect to $\bby\bl=\sum_{i=1}^n \bbx\bl_i$ is bounded
  by $\alpha/\rho$, where
  \begin{equation}
  \zeta_\eps(\rho) =\zeta_2(\rho) + 4 \gamma m^2\log(m/\beta)/\eps.
  \label{eq:zeta_hat_funct}
  \end{equation}
\end{lemma}

\begin{proof}  
  From standard bounds on the tails of the Laplace distribution (Lemma
  \ref{lem:lap_utility}), we have the following except with
  probability $\beta$:
  \[
  \max_e \left|\hat y_e - \sum_i x_{i, e}\bl \right| \leq \frac{2m}{\eps}\cdot \log\left(\frac{m}{\beta}\right)
  \]
We now condition on this level of accuracy. Since the toll function
$\tau^*(\cdot)$ is $\gamma$-Lipschitz and $\tau' = \tau^*(\bby\bl)$,
we have
  \[
  \max_e \left|\hat \tau_e - \tau'_e \right| \leq \frac{2m\gamma}{\eps}\cdot \log\left(\frac{m}{\beta}\right) \equiv \nu_\eps
  \]

Therefore a player's cost for taking the same route may increase by as
much as $m\nu_\eps$.  Further, the cost for an alternative route may
decrease by at most the same amount. Thus, each of
$\zeta_2(\rho)$-satisfied players under the flow $\bbx\bl$ in
$\Gamma^{\tau^*}$ remain $(\zeta_2(\rho)+2m\nu_\eps)$-satisfied in
game $\Gamma^{\hat \tau}$.  By Lemma~\ref{lem:fun2con}, we know that
the number of $\zeta_2(\rho)$-unsatisfied players in $\Gamma^{\tau'}$
is bounded by $n \alpha/\rho$, so the number of
$(\zeta_2(\rho)+2m\nu_\eps)$-unsatisfied players in
$\Gamma^{\hat\tau}$ is bounded by $n \alpha/\rho$ as well.
\end{proof}

We now consider what happens when instead of allowing players to best
respond given the exact congestion $\bby = \sum_{i =1}^n \bbx_{i}$, we
instead let them best respond given a private and perturbed version of
the congestion. The following general lemma will be useful, which
relates to unsatisfied players in two different congestions that are
close.  \snote{decide to prove more general lemma for other use as
  well}

\begin{lemma}
  Let $\Gamma$ be a routing game, and $\bbx$ be a flow in
  $\Gamma$. Let $\bby$ and $\bby'$ such that $\|\bby - \bby'\|_\infty
  \leq b$. Then for any number $\zeta > 0$, the set of
  $\zeta$-satisfied players in $\bbx$ with respect to $\bby$ are also
  $\zeta'$-satisfied with respect to $\bby'$, where
  \[
  \zeta' = \zeta + 2 m\gamma b.
  \]
\label{lem:transfer}
\end{lemma}

\begin{proof}
 The proof follows from the same analysis in the proof of Lemma
 \ref{lem:unsatisfied}
\end{proof}

 From the analysis of Lemma~\ref{lem:unsatisfied}, we know that
 $\|\hat\bby - \bby\bl\|_\infty \leq \frac{2m}{\eps}\cdot
 \log\left(\frac{m}{\beta}\right)$, so by instantiating
 Lemma~\ref{lem:unsatisfied} with $\rho = 2\sqrt{m\gamma n\alpha}$ and
 combining with the result of Lemma~\ref{lem:transfer}, we recover the
 bound in Lemma~\ref{lem:unsatisfied-final}, that is the number of
 $\hat\zeta_\eps$-unsatisfied players w.r.t. congestion $\hat \bby$ in
 $\bbx\bl$ and game $\Gamma^{\hat\tau}$ is bounded by $\sqrt{n\alpha
   / 4m\gamma}$ with
\[
\hat \zeta_\eps = 4\sqrt{m\gamma n \alpha} + 8\gamma m^2\log(m/\beta)/\eps.
\]

\ifnum\final=1

\fi
\fi
\end{document}